\documentclass[Journal,english]{IEEEtran}%10pt,english
\usepackage[T1]{fontenc}
\usepackage[latin9]{inputenc}
\usepackage{amsthm}
\usepackage{amsmath}
\usepackage{amssymb}
\usepackage{graphicx}
\usepackage{esint}
\usepackage{widetext}
\usepackage{flushend}

\makeatletter
  \theoremstyle{definition}
  \newtheorem{defn}{Definition}
 \theoremstyle{plain}
\newtheorem{thm}{Theorem}
  \theoremstyle{remark}
  \newtheorem{rem}{Remark}
 \ifx\proof\undefined\
   \newenvironment{proof}[1][\proofname]{\par
     \normalfont\topsep6\p@\@plus6\p@\relax
     \trivlist
     \itemindent\parindent
     \item[\hskip\labelsep
           \scshape
       #1]\ignorespaces
   }{%
     \endtrivlist\@endpefalse
   }
   \providecommand{\proofname}{Proof}
 \fi
 \theoremstyle{plain}
  \newtheorem{cor}{Corollary}
 \theoremstyle{plain}
  \newtheorem{lem}{Lemma}

\makeatother
\usepackage{float}
\usepackage[section]{placeins}

\begin{document}
\IEEEoverridecommandlockouts

\title{The Lossy Common Information of Correlated Sources}

\author{Kumar Viswanatha$^{\ast}$, Emrah Akyol$^{\dagger}$, \IEEEmembership{Member,~IEEE}, and \\ Kenneth Rose$^{\ddagger}$, \IEEEmembership{Fellow,~IEEE,}
\thanks{The work was supported by the NSF under grants CCF-1016861, CCF-1118075 and CCF-1320599. At the time of this work, all authors were with the Department of Electrical and Computer Engineering, University of California - Santa Barbara, CA. The material in this paper was presented in part at the IEEE Information Theory Workshop (ITW) at Paraty, Brazil, Oct 2011 and the IEEE International Symposium on Information Theory (ISIT) at Boston, MA, USA, Jul 2012.} 
\thanks{$^{\ast}$K. Viswanatha is currently with the Qualcomm research center, San Diego, CA, USA. (e-mail: kumar@ece.ucsb.edu)}
\thanks{$^{\dagger}$E. Akyol is with the Electrical Engineering Department, University of Southern California, CA, USA, and with the Electrical and Computer Engineering Department, University of California - Santa Barbara, CA, USA. (e-mail: eakyol@usc.edu).}
\thanks{$^{\ddagger}$K. Rose is with the Electrical and Computer Engineering Department, University of California - Santa Barbara, CA, USA. (e-mail: rose@ece.ucsb.edu)}}

%\markboth{IEEE Transactions on Information Theory,~Vol.~XX, No.~XX, XXX~XXXX}%
%{}

%\IEEEpubid{Copyright~\copyright~2013 IEEE}

\maketitle
\begin{abstract}
The two most prevalent notions of common information (CI) are due
to Wyner and Gács-Körner and both the notions can be stated as two
different characteristic points in the \textit{lossless} Gray-Wyner
region. Although the information theoretic characterizations for these
two CI quantities can be easily evaluated for random variables with
infinite entropy (eg., continuous random variables), their operational
significance is applicable only to the lossless framework. The primary
objective of this paper is to generalize these two CI notions to the
lossy Gray-Wyner network, which hence extends the theoretical foundation
to general sources and distortion measures. We begin by deriving a
single letter characterization for the lossy generalization of Wyner's
CI, defined as the minimum rate on the shared branch of the Gray-Wyner
network, maintaining minimum sum transmit rate when the two decoders
reconstruct the sources subject to individual distortion constraints.
To demonstrate its use, we compute the CI of bivariate Gaussian random
variables for the entire regime of distortions. We then similarly
generalize Gács and Körner's definition to the lossy framework. The
latter half of the paper focuses on studying the tradeoff
between the total transmit rate and receive rate in the Gray-Wyner
network. We show that this tradeoff yields a contour of points on
the surface of the Gray-Wyner region, which passes through both the
Wyner and Gács-Körner operating points, and thereby provides a unified
framework to understand the different notions of CI. We further show
that this tradeoff generalizes the two notions of CI to the excess
sum transmit rate and receive rate regimes, respectively.\end{abstract}

\section{Introduction\label{sec:Introduction}}

The quest for a meaningful and useful notion of common information
(CI) of two discrete random variables (denoted by $X$ and $Y$) has
been actively pursued by researchers in information theory for over
three decades. A seminal approach to quantify CI is due to Gács and
Körner \cite{Korner_CI} (denoted here by $C_{GK}(X,Y)$), who defined
it as the maximum amount of information relevant to both random variables,
one can extract from the knowledge of either one of them. Their result
was of considerable theoretical interest, but also fundamentally negative
in nature. They showed that $C_{GK}(X,Y)$ is usually much smaller
than the mutual information and is non-zero only when the joint distribution
satisfies certain unique properties. Wyner proposed an alternative
notion of CI \cite{Wyner_CI} (denoted here by $C_{W}(X,Y)$) inspired
by earlier work in multi-terminal source coding \cite{GW}. Wyner's
CI is defined as:

\begin{equation}
C_{W}(X,Y)=\inf I(X,Y;U)\label{eq:Wyner_CI}
\end{equation}
where the infimum is over all random variables, $U$, such that $X\leftrightarrow U\leftrightarrow Y$
form a Markov chain in that order. He showed that $C_{W}(X,Y)$ is
equal to the minimum rate on the shared branch of the lossless Gray-Wyner
network (described in section \ref{sub:Gray-Wyner-Network} and Fig.
\ref{fig:The-Gray-Wyner-network}), when the sum rate is constrained
to be the joint entropy. In other words, it is the minimum amount
of shared information that must be sent to both decoders, while restricting
the overall transmission rate to the minimum, $H(X,Y)$. 

\IEEEpubidadjcol

We note that although $C_{GK}(X,Y)$ and $C_{W}(X,Y)$ were defined
from theoretical standpoints, they play important roles in understanding
the performance limits in several practical networking and database
applications, see eg., \cite{ITW_CI}. We further note in passing
that several other definitions of CI, with applications in different
fields, have appeared in the literature \cite{dual_CI,Yamamoto_CI},
but are less relevant to us here.

\begin{figure}[!t]
\centering

\includegraphics[scale=0.4]{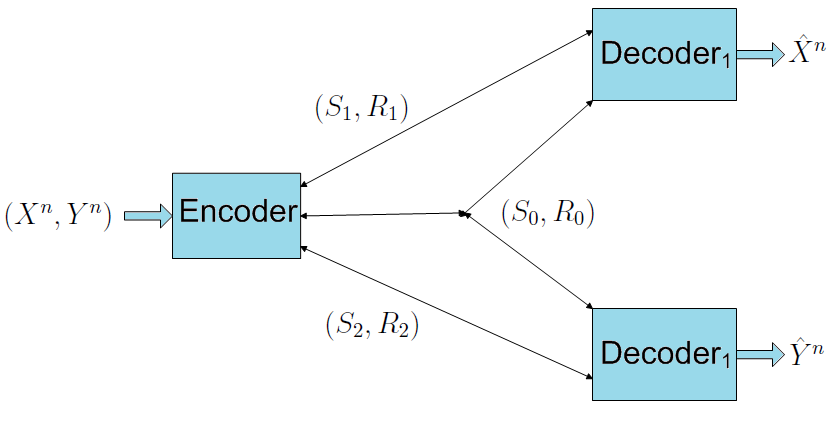}\caption{The Gray-Wyner network\label{fig:The-Gray-Wyner-network}}
\end{figure}

Although the quantity in (\ref{eq:Wyner_CI}) can be evaluated for
random variables with infinite entropy (eg. continuous random variables),
for such random variables it lacks the underlying theoretical interpretation,
namely, as a distinctive operating point in the Gray-Wyner region,
and thereby mismatches Wyner's original reasoning. This largely compromises
its practical significance and calls for a useful generalization which
can be easily extended to infinite entropy distributions. Our primary
step is to characterize a lossy coding extension of Wyner's CI (denoted
by $C_{W}(X,Y;D_{1},D_{2})$), defined as the minimum rate on the
shared branch of the Gray-Wyner network at minimum sum rate when the
sources are decoded at respective distortions of $D_{1}$ and $D_{2}$.
Note that the minimum sum rate at distortions $D_{1}$ and $D_{2}$
is given by Shannon's rate distortion function, hereafter denoted
by $R{}_{X,Y}(D_{1},D_{2})$. In this paper, our main objective is
to derive an information theoretic characterization for $C_{W}(X,Y;D_{1},D_{2})$
for general sources and distortion measures. Using this characterization,
we derive the lossy CI of two correlated Gaussian random variables
for the entire regime of distortions. This example highlights several
important characteristics of the lossy CI and challenges that underlie 
optimal encoding in the Gray-Wyner network. We note that although
there is no prior work on characterizing $C_{W}(X,Y;D_{1},D_{2})$,
in a recent work \cite{Lossy_CI_Xu}, Xu et al. provided an asymptotic
definition for $C_{W}(X,Y;D_{1},D_{2})$%
\footnote{They denote $C_{W}(X,Y;D_{1},D_{2})$ by $C_{3}(D_{1},D_{2})$%
} and showed that there exists a region of small enough distortions
where $C_{W}(X,Y;D_{1},D_{2})$ coincides with Wyner's single letter
characterization in (\ref{eq:Wyner_CI}). We further note that there have been other physical interpretations of both notions of CI, irrespective of the Gray-Wyner network, including already in \cite{Korner_CI,Wyner_CI}, whose connections with the lossy generalizations we consider herein are less direct and beyond the scope of this paper.

The last section of the paper focuses on the tradeoff between the
total transmit rate and receive rate in the Gray-Wyner network, which
directly relates the two notions of CI. Although it is well known that
the two definitions of CI can be characterized as two extreme points
in the Gray-Wyner region, no contour with operational significance
is known which connects them. We show that the tradeoff between transmit
and receive rates leads to a contour of points on the boundary of
the Gray-Wyner region, which passes through the operating points of
both Wyner and Gács-Körner. Hence, this tradeoff plays an important
role in gaining theoretical insight into more general notions of shared
information. Beyond theoretical insight, this tradeoff also plays
a role in understanding fundamental limits in many practical applications
including storage of correlated sources and minimum cost routing for
networks (see eg., \cite{ITW_CI}). Motivated by these applications,
we consider the problem of deriving a single letter characterization
for the optimal tradeoff between the total transmit versus receive
rate in the Gray-Wyner network. We provide a complete single letter
characterization for the lossless setting. We develop further insight
into this tradeoff by defining two quantities $C(X,Y;R^{'})$ and
$K(X,Y;R^{''})$, which quantify the shared rate as a function of
the total transmit and receive rates, respectively. These two quantities
generalize the lossless notions of CI to the excess sum transmit rate
and receive rate regimes, respectively. Finally, we use these properties
to derive alternate characterizations for the two definitions of CI
under a broader unified framework. We note that a preliminary version
of our results appeared in \cite{ITW_CI} and \cite{Lossy_CI}.

The rest of the paper is organized as follows. A summary of prior
results pertaining to the Gray-Wyner network and the two notions of
CI is given in Section \ref{sec:Prior-Results}. We define a lossy
generalization of Wyner's CI and derive a single letter information
theoretic characterization in Section \ref{sec:Main_result}. Next,
we specialize to the lossy CI of two correlated Gaussian random variables
using the information theoretic characterization. In Section \ref{sec:Gacs-Korner's-CI},
we extend the Gács and Körner CI definition to the lossy framework.
In Section \ref{sec:trade-off}, we study the tradeoff between the
sum rate and receive rate in the Gray-Wyner network and show that
a corresponding contour of points on the boundary of the Gray-Wyner
region emerges, which passes through both the CI operating points
of Wyner and Gács-Körner.

\section{Prior Results\label{sec:Prior-Results}}

\subsection{The Gray-Wyner Network \cite{GW}\label{sub:Gray-Wyner-Network}}

Let $(X,Y)$ be any two dependent random variables taking values in
the alphabets $\mathcal{X}$ and $\mathcal{Y}$, respectively. Let
$\hat{\mathcal{X}}$ and $\hat{\mathcal{Y}}$ be the respective reconstruction
alphabets. For any positive integer $M$, we denote the set $\{1,2\ldots M\}$
by $I_{M}$. A sequence of $n$ independent and identically distributed
(iid) random variables is denoted by $X^{n}$ and the corresponding
alphabet by $\mathcal{X}^{n}$. In what follows, for any pair of random
variables $X$ and $Y$, $R{}_{X}(\cdot)$, $R{}_{Y}(\cdot)$ and
$R{}_{X,Y}(\cdot,\cdot)$ denote the respective rate distortion functions.
With slight abuse of notation, we use $H(\cdot)$ to denote the entropy
of a discrete random variable or the differential entropy of a continuous
random variable. 

%The Gray-Wyner network, shown in Fig. \ref{fig:The-Gray-Wyner-network}, is composed of a single encoder and two decoders. The encoder observes $(X^n,Y^n)$ and transmits three indices, denoted by, $(S_0,S_1,S_2)$. The first decoder receives $(S_0,S_1)$ and reconstructs $(X^n)$ upto a distortion constraint. Similarly, the second decoder receives $(S_0,S_2)$ and reconstrucs $(Y^n)$ upto a constraint on the distortion, as shown in Fig. \ref{fig:The-Gray-Wyner-network}. 

A rate-distortion tuple $(R_{0},R_{1},R_{2},D_{1},D_{2})$ is said
to be achievable by the Gray-Wyner network if for all $\epsilon>0$,
there exists encoder and decoder mappings: 
\begin{eqnarray}
f_{E}:\mathcal{X}^{n}\times\mathcal{Y}^{n} & \rightarrow & I_{M_{0}}\times I_{M_{1}}\times I_{M_{2}}\nonumber \\
f_{D}^{(X)}:I_{M_{0}}\times I_{M_{1}} & \rightarrow & \hat{\mathcal{X}}^{n}\nonumber \\
f_{D}^{(Y)}:I_{M_{0}}\times I_{M_{2}} & \rightarrow & \mathcal{\hat{Y}}^{n}
\end{eqnarray}
where $f_{E}(X^n,Y^n)=(S_0,S_1,S_2)$, $S_0 \in I_{M_0}$, $S_1 \in I_{M_1}$ and $S_2 \in I_{M_2}$, such that the following hold:
\begin{eqnarray}
M_{i} & \leq & 2^{n(R_{i}+\epsilon)},\,\,\, i\in\{0,1,2\}\nonumber \\
\Delta_{X} & \leq & D_{1}+\epsilon\nonumber \\
\Delta_{Y} & \leq & D_{2}+\epsilon
\end{eqnarray}
where $ $$\hat{X}^{n}=f_{D}^{(X)}(S_{0},S_{1}),\,\,\hat{Y}^{n}=f_{D}^{(Y)}(S_{0},S_{2})$
and
\begin{eqnarray}
\Delta_{X} & = & \frac{1}{n}\sum_{i=1}^{n}d_{X}(X_{i},\hat{X}_{i})\nonumber \\
\Delta_{Y} & = & \frac{1}{n}\sum_{i=1}^{n}d_{Y}(Y_{i},\hat{Y}_{i})\label{eq:GW_dist_defn}
\end{eqnarray}
for some well defined single letter distortion measures $d_{X}(\cdot,\cdot)$
and $d_{Y}(\cdot,\cdot)$. The convex closure over all such achievable
rate-distortion tuples is called the achievable region for the Gray-Wyner
network. The set of all achievable rate tuples for any given distortion
$D_{1}$ and $D_{2}$ is denoted here by $\mathcal{R}_{GW}(D_{1},D_{2})$.
We denote the lossless Gray-Wyner region, as defined in \cite{GW},
simply by $\mathcal{R}_{GW}$, when $X$ and $Y$ are random variables
with finite entropy. 

Gray and Wyner \cite{GW} gave the following complete characterization
for $\mathcal{R}_{GW}(D_{1},D_{2})$. Let $(U,\hat{X},\hat{Y})$ be
any random variables jointly distributed with $(X,Y)$ and taking
values in alphabets $\mathcal{U},\hat{\mathcal{X}}$ and $\hat{\mathcal{Y}}$,
respectively (for any arbitrary $\mathcal{U}$). Let the joint density
be $P(X,Y,U,\hat{X},\hat{Y})$. All rate-distortion tuples $(R_{0},R_{1},R_{2},D_{1},D_{2})$
satisfying the following conditions are achievable:
\begin{eqnarray}
R_{0} & \geq & I(X,Y;U)\nonumber \\
R_{1} & \geq & I(X;\hat{X}|U)\nonumber \\
R_{2} & \geq & I(Y;\hat{Y}|U)\nonumber \\
D_{1} & \geq & E(d_{X}(X,\hat{X}))\nonumber \\
D_{2} & \geq & E(d_{Y}(Y,\hat{Y}))\label{eq:GW_Lossy_Region}
\end{eqnarray}
The closure of the achievable rate distortion tuples over all such
joint densities is the complete rate-distortion region for the Gray-Wyner
network, $\mathcal{R}_{GW}(D_{1},D_{2})$. For the lossless framework,
the above characterization simplifies significantly. Let $U$ be any
random variable jointly distributed with $(X,Y)$. Then, all rate
tuples satisfying the following conditions belong to the lossless
Gray-Wyner region:

\begin{eqnarray}
R_{0} & \geq & I(X,Y;U)\nonumber \\
R_{1} & \geq & H(X|U)\nonumber \\
R_{2} & \geq & H(Y|U)\label{eq:GW_Lossless_Region}
\end{eqnarray}
The convex closure of achievable rates, over all such joint densities
is denoted by $\mathcal{R}_{GW}$.

\subsection{Wyner's Common Information}

Wyner's CI, denoted by $C_{W}(X,Y)$, is defined as:
\begin{equation}
C_{W}(X,Y)=\inf I(X,Y;U)
\end{equation}
where the infimum is over all random variables $U$ such that $X\leftrightarrow U\leftrightarrow Y$
form a Markov chain in that order. Wyner showed that $C_{W}(X,Y)$
is equal to the minimum rate on the shared branch of the GW network,
while the total sum rate is constrained to be the joint entropy. To
formally state the result, we first define the set $\mathcal{R}_{W}$.
A common rate $R_{0}$ is said to belong to $\mathcal{R}_{W}$ if
for any $\epsilon>0$, there exists a point $(R_{0},R_{1},R_{2})$
such that:
\[
(R_{0},R_{1},R_{2})\in\mathcal{R}_{GW}
\]
\begin{equation}
R_{0}+R_{1}+R_{2}\leq H(X,Y)+\epsilon\label{eq:Wyner_CI_Defn}
\end{equation}
Then, Wyner showed that:
\begin{equation}
C_{W}(X,Y)=\inf R_{0}\in\mathcal{R}_{W}
\end{equation}
It is worthwhile noting that, if the random variables $(X,Y)$ are
such that every point on the plane $R_{0}+R_{1}+R_{2}=H(X,Y)$ satisfies
(\ref{eq:GW_Lossless_Region}) with equality for some joint density
$P(X,Y,U)$, then (\ref{eq:Wyner_CI_Defn}) can be simplified by setting
$\epsilon=0$, without loss in optimality. We again note that Wyner
showed that the quantity $C_{W}(X,Y)$ has other operational interpretations,
besides the Gray-Wyner network. Their relations to the lossy generalization
we define in this paper are less obvious and will be considered as
part of our future work.

\subsection{The Gács and Körner Common Information}

Let $X$ and $Y$ be two dependent random variables taking values
on finite alphabets $\mathcal{X}$ and $\mathcal{Y}$, respectively.
Let $X^{n}$ and $Y^{n}$ be $n$ independent copies of $X$ and $Y$.
Gács and Körner defined CI of two random variables as follows:
\begin{equation}
C_{GK}(X,Y)=\sup\frac{1}{n}H(f_{1}(X^{n}))\label{eq:GC_defn-1}
\end{equation}
where sup is taken over all sequences of functions $f_{1}^{(n)}$,
$f_{2}^{(n)}$, such that $P(f_{1}(X^{n})\neq f_{2}(Y^{n}))\rightarrow0$.
It can be understood as the maximum rate of the codeword that can
be generated individually at two encoders observing $X^{n}$ and $Y^{n}$
separately. To describe their main result, we need the following definition.
\begin{defn}
\label{def:J}Without loss of generality, we assume $P(X=x)>0\,\,\forall x\in\mathcal{X}$
and $P(Y=y)>0,\,\,\forall y\in\mathcal{Y}$. Ergodic decomposition
of the stochastic matrix of conditional probabilities $P(X=x|Y=y)$,
is defined by a partition of the space $\mathcal{X}\times\mathcal{Y}$
into disjoint subsets, $\mathcal{X}\times\mathcal{Y}=\bigcup_{j}\mathcal{X}_{j}\times\mathcal{Y}_{j}$,
such that $\forall j$:
\begin{eqnarray}
P(X=x|Y=y) & = & 0\,\,\,\forall x\in\mathcal{X}_{j},y\notin\mathcal{Y}_{j}\nonumber \\
P(Y=y|X=x) & = & 0\,\,\,\forall y\in\mathcal{Y}_{j},x\notin\mathcal{X}_{j}\label{eq:Erg_Decomp_Defn}
\end{eqnarray}
Observe that, the ergodic decomposition will always be such that if
$x\in\mathcal{X}_{j}$, then $Y$ must take values in $\mathcal{Y}_{j}$
and vice-versa, i.e., if $y\in\mathcal{Y}_{j}$, then $X$ must take
values only in $\mathcal{X}_{j}$. Let us define the random variable
$J$ as:
\begin{equation}
J=j\,\,\mbox{iff}\,\, x\in\mathcal{X}_{j}\Leftrightarrow y\in\mathcal{Y}_{j}\label{eq:J_defn}
\end{equation}
Gács and Körner showed that $C_{GK}(X,Y)=H(J)$. 
\end{defn}
The original definition of CI, due to Gács and Körner, was naturally
unrelated to the Gray-Wyner network, which it predates. However, an
alternate and insightful characterization of $C_{GK}(X,Y)$ in terms
of $\mathcal{R}_{GW}$, was given by Ahlswede and Körner in \cite{Ahlswede74oncommon} (and also recently appeared in \cite{dual_CI}). 
To formally state the result, we define the set $\mathcal{R}_{GK}$.
A common rate $R_{0}$ is said to belong to $\mathcal{R}_{GK}$ if
for any $\epsilon>0$, there exists a point $(R_{0},R_{1},R_{2})$
such that:

\begin{equation}
(R_{0},R_{1},R_{2})\in\mathcal{R}_{GW}\label{eq:AK_GK-1}
\end{equation}
\begin{eqnarray}
R_{0}+R_{1} & \leq & H(X)+\epsilon\nonumber \\
R_{0}+R_{2} & \leq & H(Y)+\epsilon\label{eq:constraint_AK_GK}
\end{eqnarray}
Then:
\[
C_{GK}(X,Y)=\sup R_{0}\in\mathcal{R}_{GK}
\]
Specifically, Ahlswede and Körner showed that:
\begin{eqnarray}
H(J)=C_{GK}(X,Y)=\sup I(X,Y;U)
\end{eqnarray}
subject to,
\begin{equation}
Y\leftrightarrow X\leftrightarrow U\mbox{ and }X\leftrightarrow Y\leftrightarrow U\label{eq:GK_AK_mc}
\end{equation}
This characterization of $C_{GK}(X,Y)$ in terms of the Gray-Wyner
network offers a crisp understanding of the difference in objective
between the two approaches. It will become evident in Section \ref{sec:trade-off}
that the two are in fact instances of a more general objective, within
a broader unified framework. Again it is worthwhile noting that, if
every point on the intersection of the planes $R_{0}+R_{1}=H(X)$
and $R_{0}+R_{2}=H(Y)$ satisfies (\ref{eq:GW_Lossless_Region}) with
equality for some joint density $P(X,Y,U)$, then the above definition
can be simplified by setting $\epsilon=0$.

\section{Lossy Extension of Wyner's Common Information \label{sec:Main_result}}

\subsection{Definition}

We generalize Wyner's CI to the lossy framework and denote it by $C_{W}(X,Y;D_{1},D_{2})$.
Let $\mathcal{R}_{W}(D_{1},D_{2})$ be the set of all $R_{0}$ such
that, for any $\epsilon>0$, there exists a point $(R_{0},R_{1},R_{2})$
satisfying the following conditions:
\[
(R_{0},R_{1},R_{2})\in\mathcal{R}_{GW}(D_{1},D_{2})
\]
\[
R_{0}+R_{1}+R_{2}\leq R_{X,Y}(D_{1},D_{2})+\epsilon
\]
Then, the lossy generalization of Wyner's CI is defined as the infimum
over all such shared rates $R_{0}$, i.e., 
\[
C_{W}(X,Y;D_{1},D_{2})=\inf R_{0}\in\mathcal{R}_{W}(D_{1},D_{2})
\]
Note that, for any distortion pair $(D_{1},D_{2})$, if every point
on the plane $R_{0}+R_{1}+R_{2}=R_{X,Y}(D_{1},D_{2})$ satisfies (\ref{eq:GW_Lossy_Region})
with equality for some joint density $P(X,Y,U,\hat{X},\hat{Y})$,
then the above definition can be simplified by setting $\epsilon=0$.
We note that the plane $R_{0}+R_{1}+R_{2}=R{}_{X,Y}(D_{1},D_{2})$
is called the Pangloss plane in the literature \cite{GW}. We 
note that the above operational definition of $C_{W}(X,Y;D_{1},D_{2})$,
has also appeared recently in \cite{Lossy_CI_Xu}, albeit without
a single letter information theoretic characterization. The primary
objective of Section \ref{sub:Single-Letter-Characterization} is
to characterize $C_{W}(X,Y;D_{1},D_{2})$ for general sources and
distortion measures. Wyner gave the complete single letter characterization
of $C_{W}(X,Y;0,0)$ when $X$ and $Y$ have finite joint entropy
and the distortion measure is Hamming distortion ($d(x,\hat{x})=\boldsymbol{1}_{x\neq\hat{x}}$,
$d(y,\hat{y})=\boldsymbol{1}_{y\neq\hat{y}}$), i.e., 
\begin{equation}
C_{W}(X,Y)=C_{W}(X,Y;0,0)=\inf I(X,Y;U)
\end{equation}
where the infimum is over all $U$ satisfying $X\leftrightarrow U\leftrightarrow Y$.

\subsection{Single Letter Characterization of $C_{W}(X,Y;D_{1},D_{2})$\label{sub:Single-Letter-Characterization}}

To simplify the exposition and to make the proof more intuitive, in
the following theorem, we assume that for any distortion pair $(D_{1},D_{2})$,
every point on the Pangloss plane (i.e., $R_{0}+R_{1}+R_{2}=R_{X,Y}(D_{1},D_{2})$)
satisfies (\ref{eq:GW_Lossy_Region}) with equality for some joint
density $P(X,Y,U,\hat{X},\hat{Y})$. We handle the more general setting
in Appendix A. 
\begin{thm}
\label{thm:main}A single letter characterization of $C_{W}(X,Y;D_{1},D_{2})$
is given by:
\begin{equation}
C_{W}(X,Y;D_{1},D_{2})=\inf I(X,Y;U)\label{eq:mt_1}
\end{equation}
where the infimum is over all joint densities $P(X,Y,\hat{X},\hat{Y},U)$
such that the following Markov conditions hold:
\begin{eqnarray}
\hat{X}\leftrightarrow & U & \leftrightarrow\hat{Y}\label{eq:mt_2}\\
(X,Y)\leftrightarrow & (\hat{X},\hat{Y}) & \leftrightarrow U\label{eq:mt_3}
\end{eqnarray}
and where $P(\hat{X},\hat{Y}|X,Y)\in\mathcal{P}_{D_{1},D_{2}}^{X,Y}$
is any joint distribution which achieves the rate distortion function
at $(D_{1},D_{2})$, i.e., $I(X,Y;\hat{X},\hat{Y})=R_{X,Y}(D_{1},D_{2})$,
$E(d_{X}(X,\hat{X}))\leq D_{1}$ and $E(d_{Y}(Y,\hat{Y}))\leq D_{2}$,
$\forall P(\hat{X},\hat{Y}|X,Y)\in\mathcal{P}_{D_{1},D_{2}}^{X,Y}$. \end{thm}
\begin{rem}
If we set $\hat{\mathcal{X}}=\mathcal{X}$, $\hat{\mathcal{Y}}=\mathcal{Y}$
and consider the Hamming distortion measure, at $(D_{1},D_{2})=(0,0)$,
it is easy to show that Wyner's CI is obtained as a special case,
i.e., $C_{W}(X,Y;0,0)=C_{W}(X,Y)$. \end{rem}
\begin{proof}
We note that, although there are arguably simpler methods to prove
this theorem, we choose the following approach as it uses only the
Gray-Wyner theorem without recourse to any supplementary results.
Further, we assume that there exists a unique encoder $P^{*}(\hat{X}^{*},\hat{Y}^{*}|X,Y)\in\mathcal{P}_{D_{1},D_{2}}^{X,Y}$
which achieves $R{}_{X,Y}(D_{1},D_{2})$. The proof of the theorem
when there are multiple encoders in $\mathcal{P}_{D_{1},D_{2}}^{X,Y}$
follows directly. 

Our objective is to show that every point in the intersection of $\mathcal{R}_{GW}(D_{1},D_{2})$
and the Pangloss plane has $R_{0}=I(X,Y;U)$ for some $U$ jointly
distributed with $(X,Y,\hat{X}^{*},\hat{Y}^{*})$ and satisfying conditions
(\ref{eq:mt_2}) and (\ref{eq:mt_3}). We first prove that every point
in the intersection of the Pangloss plane and $\mathcal{R}_{GW}(D_{1},D_{2})$
is achieved by a joint density satisfying (\ref{eq:mt_2}) and (\ref{eq:mt_3}).
Towards showing this, we begin with an alternate characterization
of $\mathcal{R}_{GW}(D_{1},D_{2})$ (which is also complete) due to
Venkataramani et al. (see section III.B in \cite{VKG})%
\footnote{We note that the theorem can be proved even using the original Gray-Wyner
characterization. However, if we begin with that characterization,
we would require the random variables to satisfy two additional Markov
conditions beyond (\ref{eq:mt_2}) and (\ref{eq:mt_3}). These Markov
conditions can in fact be shown to be redundant from the Kuhn-Tucker
conditions. The alternate approach we choose circumvents these supplementary
arguments. %
}. Let $(U,\hat{X},\hat{Y})$ be any random variables jointly distributed
with $(X,Y)$ such that $E(d_{X}(X,\hat{X}))\leq D_{1}$ and $E(d_{Y}(Y,\hat{Y}))\leq D_{2}$.
Then any rate tuple $(R_{0},R_{1},R_{2})$ satisfying the following
conditions belongs to $\mathcal{R}_{GW}(D_{1},D_{2})$:
\begin{eqnarray}
R_{0} & \geq & I(X,Y;U)\nonumber \\
R_{1}+R_{0} & \geq & I(X,Y;U,\hat{X})\nonumber \\
R_{2}+R_{0} & \geq & I(X,Y;U,\hat{Y})\nonumber \\
R_{0}+R_{1}+R_{2} & \geq & I(X,Y;U,\hat{X},\hat{Y})+I(\hat{X};\hat{Y}|U)\label{eq:VKG_characterization}
\end{eqnarray}
It is easy to show that the above characterization is equivalent to
(\ref{eq:GW_Lossy_Region}). As the above characterization is complete,
this implies that, if a rate-distortion tuple $(R_{0},R_{1},R_{2},D_{1},D_{2})$
is achievable for the Gray-Wyner network, then we can always find
random variables $(U,\hat{X},\hat{Y})$ such that $E(d_{X}(X,\hat{X}))\leq D_{1}$,
$E(d_{Y}(Y,\hat{Y}))\leq D_{2}$ and satisfying (\ref{eq:VKG_characterization}).
We are further interested in characterizing the points in $\mathcal{R}_{GW}(D_{1},D_{2})$
that lie on the Pangloss plane, i.e., $R_{0}+R_{1}+R_{2}=R{}_{X,Y}(D_{1},D_{2})$.
Therefore, for any rate tuple $(R_{0},R_{1},R_{2})$ on the Pangloss
plane in $\mathcal{R}_{GW}(D_{1},D_{2})$, we have the following series
of inequalities:

\begin{eqnarray}
R{}_{X,Y}(D_{1},D_{2}) & = & R_{0}+R_{1}+R_{2}\nonumber \\
 & \geq & I(X,Y;U,\hat{X},\hat{Y})+I(\hat{X};\hat{Y}|U)\nonumber \\
 & \geq & I(X,Y;\hat{X},\hat{Y})+I(\hat{X};\hat{Y}|U)\nonumber \\
 & \geq^{(a)} & R{}_{X,Y}(D_{1},D_{2})+I(\hat{X};\hat{Y}|U)\nonumber \\
 & \geq & R{}_{X,Y}(D_{1},D_{2})
\end{eqnarray}
where $(a)$ follows because $(\hat{X},\hat{Y})$ satisfy the distortion
constraints. Since the above chain of inequalities start and end with
the same quantity, they must all be identities and we have:
\begin{eqnarray}
I(\hat{X};\hat{Y}|U) & = & 0\nonumber \\
I(X,Y;U,\hat{X},\hat{Y}) & = & I(X,Y;\hat{X},\hat{Y})\nonumber \\
I(X,Y;\hat{X},\hat{Y}) & = & R{}_{X,Y}(D_{1},D_{2})
\end{eqnarray}
By assumption, there is a unique encoder, $P^{*}(\hat{X}^{*},\hat{Y}^{*}|X,Y)$,
which achieves $I(X,Y;\hat{X},\hat{Y})=R{}_{X,Y}(D_{1},D_{2})$. It
therefore follows that every point in $\mathcal{R}_{GW}(D_{1},D_{2})$
that lies on the Pangloss plane satisfies (\ref{eq:VKG_characterization})
for some joint density satisfying (\ref{eq:mt_2}) and (\ref{eq:mt_3}). 

It remains to be shown that any joint density $(X,Y,\hat{X}^{*},\hat{Y}^{*},U)$
satisfying (\ref{eq:mt_2}) and (\ref{eq:mt_3}) leads to a sub-region
of $\mathcal{R}_{GW}(D_{1},D_{2})$ which has at least one point on
the Pangloss plane with $R_{0}=I(X,Y;U)$. Formally, denote by $\mathcal{R}(U)$,
the region (\ref{eq:VKG_characterization}) achieved by a joint density
$(X,Y,\hat{X}^{*},\hat{Y}^{*},U)$ satisfying (\ref{eq:mt_2}) and
(\ref{eq:mt_3}). We need to show that $\exists(R_{0},R_{1},R_{2})\in\mathcal{R}(U)$
such that:
\begin{eqnarray}
R_{0}+R_{1}+R_{2} & = & R{}_{X,Y}(D_{1},D_{2})\nonumber \\
R_{0} & = & I(X,Y;U)\label{eq:achievability_to_show}
\end{eqnarray}
Consider the point, $(R_{0},R_{1},R_{2})=\left(I(X,Y;U),I(X,Y;\hat{X}^{*}|U),I(X,Y;\hat{Y}^{*}|U,\hat{X}^{*})\right)$
for any joint density $(X,Y,\hat{X}^{*},\hat{Y}^{*},U)$ satisfying
(\ref{eq:mt_2}) and (\ref{eq:mt_3}). Clearly the point satisfies
the first two conditions in (\ref{eq:VKG_characterization}). Next,
we note that:
\begin{eqnarray}
R_{0}+R_{2} & = & I(X,Y;U)+I(X,Y;\hat{Y}^{*}|U,\hat{X}^{*})\nonumber \\
 & \geq^{(b)} & I(X,Y;U)+I(X,Y;\hat{Y}^{*}|U)\nonumber \\
 & = & I(X,Y;\hat{Y}^{*},U)
\end{eqnarray}
\begin{eqnarray}
R_{0}+R_{1}+R_{2} & =^{(c)} & I(X,Y;\hat{X}^{*},\hat{Y}^{*},U)\\
 & = & I(X,Y;\hat{X}^{*},\hat{Y}^{*})=R{}_{X,Y}(D_{1},D_{2})\nonumber 
\end{eqnarray}
where $(b)$ and $(c)$ follow from the fact that the joint density
satisfies (\ref{eq:mt_2}) and (\ref{eq:mt_3}). Hence, we have shown
the existence of one point in $\mathcal{R}(U)$ satisfying (\ref{eq:achievability_to_show})
for every joint density $(X,Y,\hat{X}^{*},\hat{Y}^{*},U)$ satisfying
(\ref{eq:mt_2}) and (\ref{eq:mt_3}), proving the theorem.
\end{proof}
The following corollary sheds light on several properties related
to the optimizing random variables $U$ in Theorem \ref{thm:main}.
These properties significantly simplify the computation of lossy CI.
\begin{cor}
The joint distribution that optimizes (\ref{eq:mt_1}) in Theorem
\ref{thm:main} satisfies the following properties:
\begin{eqnarray}
\hat{X}^{*}\leftrightarrow & (X,Y,U) & \leftrightarrow\hat{Y}^{*}\nonumber \\
\hat{X}^{*}\leftrightarrow & (X,U) & \leftrightarrow Y\nonumber \\
\hat{Y}^{*}\leftrightarrow & (Y,U) & \leftrightarrow X\label{eq:Lem_add_Mark_Prop}
\end{eqnarray}
i.e., the conditional density can always be written as:
\vfill\eject
\begin{equation}
P(\hat{X}^{*},\hat{Y}^{*},U|X,Y)=P(\hat{X}^{*},U|X)P(\hat{Y}^{*},U|Y)
\end{equation}
\end{cor}
\begin{proof}

We relegate the proof to Appendix B as it is quite orthogonal to the
main flow of the paper. 
\end{proof}

We note that, in general, $C_{W}(X,Y;D_{1},D_{2})$ is neither convex/concave
nor monotonic with respect to $(D_{1},D_{2})$. As we will see later,
$C_{W}(X,Y;D_{1},D_{2})$ is non-monotonic even for two correlated
Gaussian random variables under mean squared distortion measure. This
makes it hard to establish conclusive inequality relations between
$C_{W}(X,Y;D_{1},D_{2})$ and $C_{W}(X,Y)$ for all distortions. However,
in the following lemma, we establish sufficient conditions on $(D_{1},D_{2})$
for $C_{W}(X,Y;D_{1},D_{2})\lessgtr C_{W}(X,Y)$. In Appendix C, we 
review some of the results pertinent to Shannon lower bounds for vectors
of random variables, that will be useful in the following Lemma. 
\begin{lem}
\label{lem:ESLB}For any pair of random variables $(X,Y)$, \end{lem}
\begin{itemize}
\item (i) $C_{W}(X,Y;D_{1},D_{2})\leq C_{W}(X,Y)$ at $(D_{1},D_{2})$ if
$\exists(\tilde{D}_{1},\tilde{D}_{2})$ such that $\tilde{D}_{1}\leq D_{1}$,
$\tilde{D}_{2}\leq D_{2}$ and $R{}_{X,Y}(\tilde{D}_{1},\tilde{D}_{2})=C_{W}(X,Y)$.
\item (ii) For any difference distortion measures, $C_{W}(X,Y;D_{1},D_{2})\geq C_{W}(X,Y)$,
if the Shannon lower bound for $R{}_{XY}(D_{1},D_{2})$ is tight at
$(D_{1},D_{2})$.\end{itemize}
\begin{proof}
The proof of (i) is straightforward and hence omitted. Towards proving
(ii), we appeal to standard techniques \cite{Berger_book,cond_rd}
(also refer to Appendix C) which immediately show that the conditional
distribution $P(\hat{X}^{*},\hat{Y}^{*}|X,Y)$ that achieves $R{}_{X,Y}(D_{1},D_{2})$
when Shannon lower bound is tight, has independent backward channels,
i.e.:
\begin{equation}
P_{X,Y|\hat{X}^{*},\hat{Y}^{*}}(x,y|\hat{x}^{*},\hat{y}^{*})=Q_{X|\hat{X}^{*}}(x|\hat{x}^{*})Q_{Y|\hat{Y}^{*}}(y|\hat{y}^{*})\nonumber\label{eq:slb_cond-1}
\end{equation}

Let us consider any $U$ that satisfies $(\hat{X}^{*}\leftrightarrow U\leftrightarrow\hat{Y}^{*})$
and $(X,Y)\leftrightarrow(\hat{X}^{*},\hat{Y}^{*})\leftrightarrow U$.
It is easy to verify that any such joint density also satisfies $X\leftrightarrow U\leftrightarrow Y$.
As the infimum for $C_{W}(X,Y)$ is taken over a larger set of joint
densities, we have $C_{W}(X,Y;D_{1},D_{2})\geq C_{W}(X,Y)$.
\end{proof}
The above lemma highlights the anomalous behavior of $C_{W}(X,Y;D_{1},D_{2})$
with respect to the distortions. Determining the conditions for equality
in Lemma \ref{lem:ESLB}.(ii) is an interesting problem in its own
right. It was shown in \cite{Lossy_CI_Xu} that there always exists
a region of distortions around the origin such that $C_{W}(X,Y;D_{1},D_{2})=C_{W}(X,Y)$.
We will further explore the underlying connections between these results
as part of future work.

\subsection{Bivariate Gaussian Example}

Let $X$ and $Y$ be jointly Gaussian random variables with zero mean,
unit variance and a correlation coefficient of $\rho$. Let the distortion measure be the mean squared error (MSE), i.e., $D_{1}(X,\hat{X})=(X-\hat{X})^{2}$ and $D_{2}(Y,\hat{Y})=(Y-\hat{Y})^{2}$.

\begin{figure*}
\begin{equation}
R{}_{X,Y}(D_{1},D_{2})=\begin{cases}
\frac{1}{2}\log\left(\frac{1-\rho^{2}}{D_{1}D_{2}}\right) & \mbox{if }\bar{D}_{1}\bar{D}_{2}\geq\rho^{2}\\
\frac{1}{2}\log\left(\frac{1-\rho^{2}}{D_{1}D_{2}-\left(\rho-\sqrt{\bar{D}_{1}\bar{D}_{2}}\right)^{2}}\right) & \mbox{if }\bar{D}_{1}\bar{D}_{2}\leq\rho^{2},\min\left\{ \frac{\bar{D}_{1}}{\bar{D}_{2}},\frac{\bar{D}_{2}}{\bar{D}_{1}}\right\} \geq\rho^{2}\\
\frac{1}{2}\log\left(\frac{1}{\min(D_{1},D_{2})}\right) & \mbox{if }\bar{D}_{1}\bar{D}_{2}\leq\rho^{2},\min\left\{ \frac{\bar{D}_{1}}{\bar{D}_{2}},\frac{\bar{D}_{2}}{\bar{D}_{1}}\right\} <\rho^{2}
\end{cases}\label{eq:RD_Func_Gauss}
%\tag{20}
\end{equation}

\begin{eqnarray}
C_{W}(X,Y,D_{1},D_{2}) & = & \begin{cases}
C_{W}(X,Y)=\frac{1}{2}\log\frac{1+\rho}{1-\rho} & \mbox{if \ensuremath{\max\{D_{1},D_{2}\}\leq1-\rho}}\\
R_{X,Y}(D_{1},D_{2}) & \mbox{if }\bar{D}_{1}\bar{D}_{2}\leq\rho^{2}\\
\frac{1}{2}\log\frac{1-\rho^{2}}{\left(1-\frac{\rho^{2}}{\bar{D}_{1}}\right)D_{1}} & \mbox{if }D_{1}>1-\rho,\bar{D}_{1}\bar{D}_{2}>\rho^{2}\\
\frac{1}{2}\log\frac{1-\rho^{2}}{\left(1-\frac{\rho^{2}}{\bar{D}_{2}}\right)D_{2}} & \mbox{if }D_{2}>1-\rho,\bar{D}_{1}\bar{D}_{2}>\rho^{2}
\end{cases}\label{eq:Gauss_Lossy_CI}
\end{eqnarray}

\end{figure*}

%\begin{figure*}[b!]
%\end{figure*}

\vfill\eject

Hereafter, for simplicity, we assume that $\rho \in [0,1]$, noting that all results can be easily extended to negative values of $\rho$ with appropriate modifications. The joint rate distortion function is given by (\ref{eq:RD_Func_Gauss}) at the top of the page (see \cite{Berger_book}), where $\bar{D}_{i}=1-D_{i}$.

Let us first consider the range of distortions such that $\bar{D}_{1}\bar{D}_{2}\geq\rho^{2}$.
The RD-optimal random encoder is such that $P(X|\hat{X}^{*})$ and
$P(Y|\hat{Y}^{*})$ are two independent zero mean Gaussian channels
with variances $D_{1}$ and $D_{2}$, respectively. It is easy to
verify that the optimal reproduction distribution (for $(\hat{X}^{*},\hat{Y}^{*})$)
is jointly Gaussian with zero mean. The covariance matrix for $(\hat{X}^{*},\hat{Y}^{*})$
is: 
\begin{equation}
\Sigma_{\hat{X}^{*}\hat{Y}^{*}}=\left[\begin{array}{cc}
\bar{D}_{1} & \rho\\
\rho & \bar{D}_{2}
\end{array}\right]\label{eq:Xs_Ys_Cov_Mat}
\end{equation}

Observe that, at these distortions, the Shannon lower bound is tight.
Next, let us consider the range of distortions such that $\bar{D}_{1}\bar{D}_{2}\leq\rho^{2}$.
The RD-optimal random encoder in this distortion range is such that
$\hat{X}^{*}=a\hat{Y}^{*}$ and the conditional distribution $P(X,Y|\hat{X}^{*}=a\hat{Y}^{*})$
is jointly Gaussian with correlated components, where: 
\begin{equation}
a=\begin{cases}
\sqrt{\frac{\bar{D}_{1}}{\bar{D}_{2}}} & \min\left\{ \frac{\bar{D}_{1}}{\bar{D}_{2}},\frac{\bar{D}_{2}}{\bar{D}_{1}}\right\} \geq\rho^{2}\\
\rho & \min\left\{ \frac{\bar{D}_{1}}{\bar{D}_{2}},\frac{\bar{D}_{2}}{\bar{D}_{1}}\right\} <\rho^{2}
\end{cases}
\end{equation}

Clearly, the Shannon lower bound is not tight in this regime of distortions.
Note that the regime of distortions where $\bar{D}_{1}\bar{D}_{2}\leq\rho^{2}$
and $\min\left\{ \frac{\bar{D}_{1}}{\bar{D}_{2}},\frac{\bar{D}_{2}}{\bar{D}_{1}}\right\} <\rho^{2}$,
is degenerate in the sense that one of the two distortions can be
reduced without incurring any excess sum-rate \cite{Jayanth_Vector_SR}. 

It was shown in \cite{Lossy_CI_Xu} that, for two correlated Gaussian
random variables, $C_{W}(X,Y)=\frac{1}{2}\log\frac{1+\rho}{1-\rho}$
and the infimum achieving $U^{*}$ is a standard Gaussian random variable
jointly distributed with $(X,Y)$ as: 
\begin{eqnarray}
X & = & \sqrt{\rho}U^{*}+\sqrt{1-\rho}N_{1}\nonumber \\
Y & = & \sqrt{\rho}U^{*}+\sqrt{1-\rho}N_{2}\label{eq:Gauss_Lossless_U}
\end{eqnarray}
where $N_{1}$, $N_{2}$ and $U^{*}$ are independent standard Gaussian
random variables. Equipped with these results, we derive in the following
theorem, the lossy CI of two correlated Gaussian random variables,
and then demonstrate its anomalous behavior.

\begin{thm}
The lossy CI of two correlated zero-mean Gaussian random variables
with unit variance and correlation coefficient $\rho$, is given by
(\ref{eq:Gauss_Lossy_CI}) at the top of the page. 
\end{thm}

\begin{rem}
The different regimes of distortions indicated in (\ref{eq:Gauss_Lossy_CI})
are depicted in Fig. \ref{fig:curve}(a). The lossy CI is equal to
the corresponding lossless information theoretic characterization
in the regime where both $D_{1}$ and $D_{2}$ are smaller than $1-\rho$.
In the regime where $\bar{D}_{1}\bar{D}_{2}\leq\rho^{2}$, the lossy
CI is equal to the joint rate-distortion function, i.e., all the bits
are sent on the shared branch. In the other two regimes, the lossy
CI is strictly greater than the lossless characterization, i.e., $C_{W}(X,Y;D_{1},D_{2})>C_{W}(X,Y)$.
To illustrate, we fix $\rho=0.5$ and $D_{2}=0.2$, and plot the $C_{W}(X,Y;0.2,D_{2})$
as a function of $D_{1}$, as shown in Fig. \ref{fig:curve}(b). Observe
that the lossy CI remains a constant till point $A$, where $D_{1}=1-\rho$.
It is strictly greater than the lossless characterization between
points $A$ and $B$, i.e., between $D_{1}=1-\rho$ and $D_{1}=1-\frac{\rho}{1-D_{1}}$
and finally is equal to the joint rate-distortion function for all
points to the right of $B$. This result is quite counter-intuitive
and reveals a surprising property of the Gray-Wyner network. Note
that while traversing from the origin to point $A$, the lossy CI
is a constant, and decreasing only the rates of the side branches
is optimal in achieving the minimum sum rate at the respective distortions.
However, between points $A$ and $B$, the rate on the common branch
\textit{increases}, while the rates on the side branches continue
to decrease in order to maintain sum rate optimality at the respective
distortions. This implies that, even a Gaussian source under mean
squared error distortion measure, one of the simplest successively
refinable examples in point to point settings, is not successively
refinable on the Gray-Wyner network, i.e, the bit streams on the three
branches, when $D_{1}$ is set to a value in between $(A,B)$, are
not subsets of the respective bit-streams required to achieve a distortion
in between the origin and point $A$. Finally note that, for all points
to the right of $B$, all the information is carried on the common
branch and the side branches are left unused, i.e., to achieve minimum
sum rate, all the bits must be transmitted on the shared branch. This
example clearly demonstrates that the lossy CI is neither convex/concave
nor is monotonic in general. 

Before stating the formal proof, we provide a high level intuitive
argument to justify the non-monotone behavior of $C_{W}(X,Y;D_{1},D_{2})$.
First, it follows from results in \cite{Lossy_CI_Xu} that there exists
a region of distortions around origin such that $C_{W}(X,Y;D_{1},D_{2})=C_{W}(X,Y)$,
where the optimizing $U$ in Theorem \ref{thm:main} is equal to $U^{*}$.
Next, recall that, if $\bar{D}_{1}\bar{D}_{2}\geq\rho^{2}$, Shannon
lower bound is tight and hence from Lemma \ref{lem:ESLB}, $C_{W}(X,Y;D_{1},D_{2})\geq C_{W}(X,Y)$.
For $C_{W}(X,Y;D_{1},D_{2})$ to be equal to $C_{W}(X,Y)$, there
should exist a joint density over $(X,Y,\hat{X}^{*},\hat{Y}^{*},U^{*})$,
such that, $(X,Y,U^{*})$ is distributed according to (\ref{eq:Gauss_Lossless_U})
and the Markov chain condition $X\leftrightarrow\hat{X}^{*}\leftrightarrow U^{*}\leftrightarrow\hat{Y}^{*}\leftrightarrow Y$
is satisfied. However, if $D_{1}>1-\rho$ and $\bar{D}_{1}\bar{D}_{2}\geq\rho^{2}$,
$I(X;\hat{X}^{*})=\frac{1}{2}\log\frac{1}{D_{1}}<\frac{1}{2}\log\frac{1}{1-\rho}=I(X;U^{*})$.
Therefore, it is impossible to find such a joint density and hence
$C_{W}(X,Y;D_{1},D_{2})>C_{W}(X,Y)$. Finally, note that if $\bar{D}_{1}\bar{D}_{2}\leq\rho^{2}$,
$\hat{X}^{*}=a\hat{Y}^{*}$ and hence $U^{*}$ must be equal to $X^{*}$
to satisfy $\hat{X}^{*}\leftrightarrow U^{*}\leftrightarrow\hat{Y}^{*}$.
This implies that $C_{W}(X,Y;D_{1},D_{2})=R_{X,Y}(D_{1},D_{2})$,
which is a monotonically decreasing function. The above arguments
clearly demonstrate the non-monotone behavior of $C_{W}(X,Y;D_{1},D_{2})$. 
\end{rem}
\begin{figure*}[!t]
\centering\includegraphics[scale=0.22]{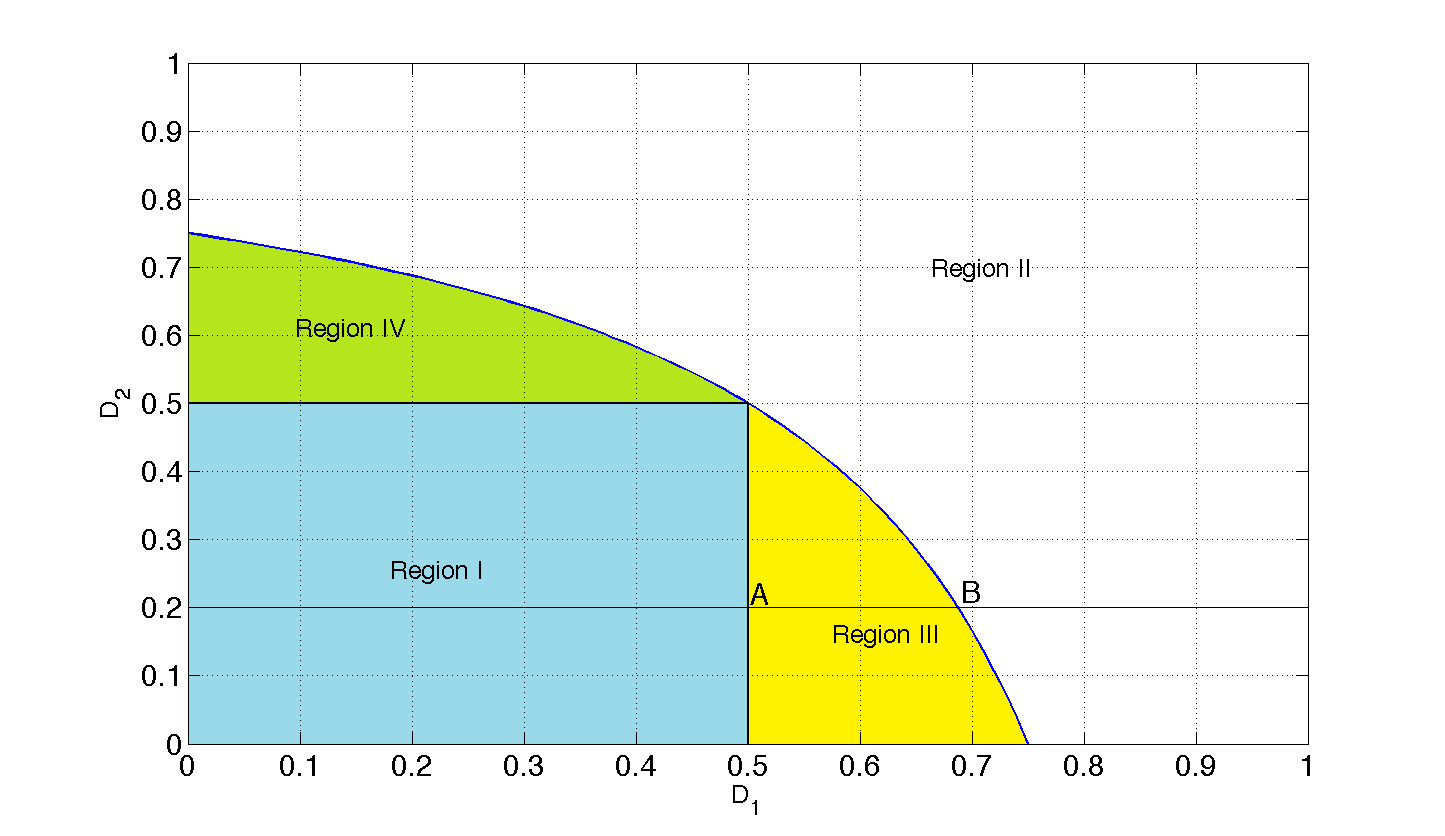}\includegraphics[scale=0.21]{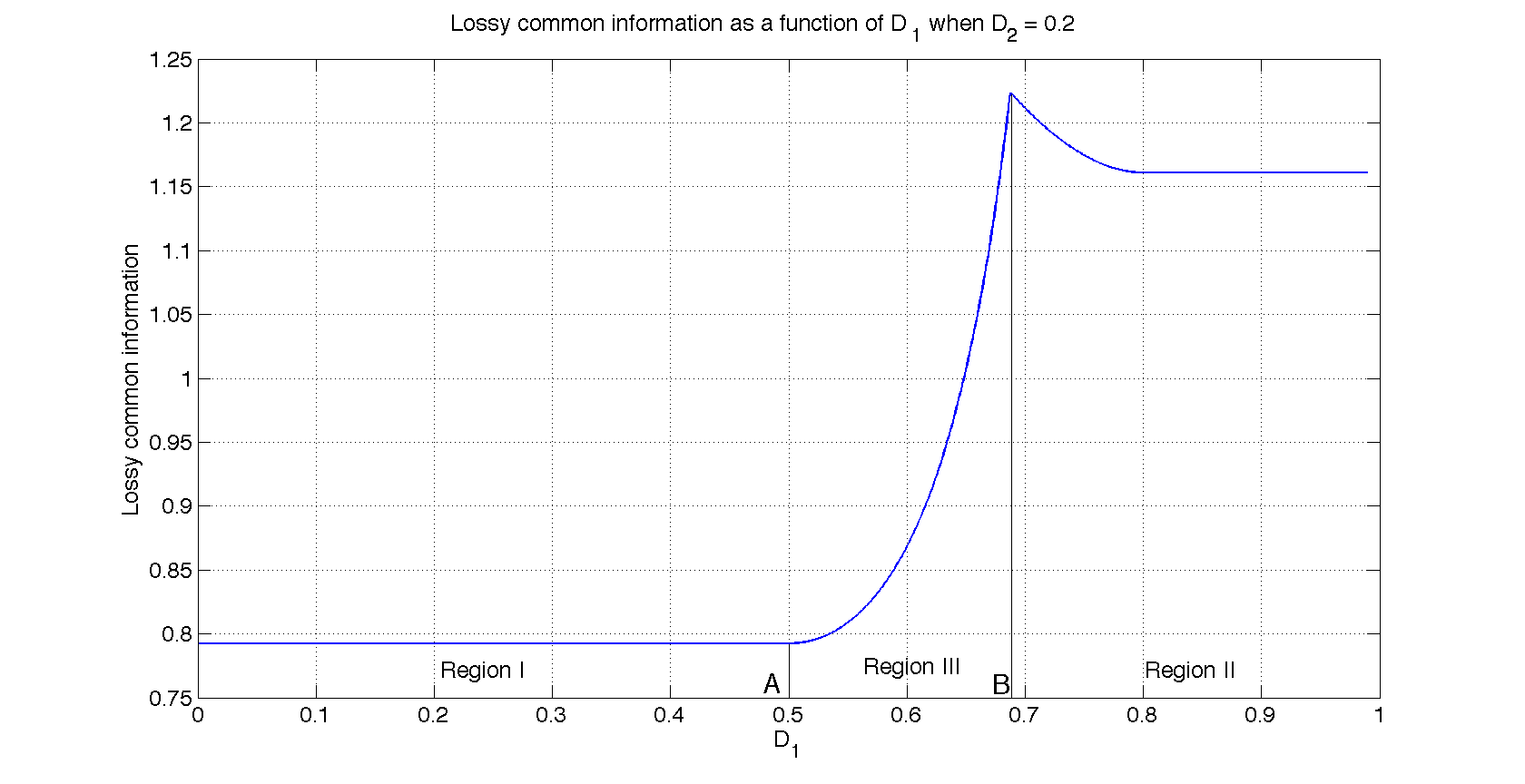}\caption{The figure on the left shows the different regions of distortions 
indicated in (\ref{eq:Gauss_Lossy_CI}) when $\rho=0.5$. Region I is blue, region II is white, region III is yellow and region IV is green. The figure
on the right shows the lossy CI for two correlated Gaussian random
variables at a fixed $D_{2}=0.2$ and as a function of $D_{1}$. Observe
that the curve is equal to the lossless CI till $D_{1}=1-\rho$, then
increases till $D_{1}=1-\frac{\rho^{2}}{1-D_{2}}$ and then finally
drops off as $R_{X,Y}(D_{1},D_{2})$. \label{fig:curve}}
\end{figure*}

\begin{proof}
We first consider the regime of distortions where $\max\{D_{1},D_{2}\}\leq1-\rho$.
At these distortions, the Shannon lower bound is tight and hence by
Lemma \ref{lem:ESLB}, $C_{W}(X,Y;D_{1},D_{2})\geq C_{W}(X,Y)$. To
prove that $C_{W}(X,Y,D_{1},D_{2})=C_{W}(X,Y)$, it is sufficient
for us to show the existence of a joint distribution over $(X,Y,\hat{X}^{*},\hat{Y}^{*},U^{*})$
satisfying (\ref{eq:mt_2}) and (\ref{eq:mt_3}), where $(X,Y,U^{*})$
satisfy (\ref{eq:Gauss_Lossless_U}) and $(\hat{X}^{*},\hat{Y}^{*})$
achieve joint RD optimality at $(D_{1},D_{2})$. We can generate $(\hat{X}^{*},\hat{Y}^{*})$
by passing $U^{*}$ through independent Gaussian channels as follows:

\begin{eqnarray}
\hat{X}^{*} & = & \sqrt{\rho}U^{*}+\sqrt{1-D_{1}-\rho}\tilde{N}_{1}\nonumber \\
\hat{Y}^{*} & = & \sqrt{\rho}U^{*}+\sqrt{1-D_{2}-\rho}\tilde{N}_{2}
\end{eqnarray}
where $\tilde{N}_{1}$ and $\tilde{N}_{2}$ are independent standard
Gaussian random variables independent of both $N_{1}$ and $N_{2}$.
Therefore there exists a joint density over $(X,Y,\hat{X}^{*},\hat{Y}^{*},U^{*})$
satisfying $\hat{X}^{*}\leftrightarrow U^{*}\leftrightarrow\hat{Y}^{*}$
and $(X,Y)\leftrightarrow(\hat{X}^{*},\hat{Y}^{*})\leftrightarrow U^{*}$.
This shows that $C_{W}(X,Y;D_{1},D_{2})\leq C_{W}(X,Y)$. Therefore
in the range $\max\{D_{1},D_{2}\}\leq1-\rho$, we have $C_{W}(X,Y;D_{1},D_{2})=C_{W}(X,Y)$.
We note that, for the symmetric case with $D_{1}=D_{2}\leq1-\rho$,
this specific result was already deduced in \cite{Lossy_CI_Xu}, albeit
using an alternate approach deriving results from conditional rate-distortion
theory. 

We next consider the range of distortions where $\bar{D}_{1}\bar{D}_{2}\leq\rho^{2}$.
Note that the Shannon lower bound for $R{}_{X,Y}(D_{1},D_{2})$ is
not tight in this range. However, the RD-optimal conditional distribution
$P(\hat{X}^{*},\hat{Y}^{*}|X,Y)$ in this distortion range is such
that $\hat{X}^{*}=a\hat{Y}^{*}$, for some constant $a$. Therefore
the only $U$ that satisfies $(\hat{X}^{*}\leftrightarrow U\leftrightarrow\hat{Y}^{*})$
is $U=\hat{X}^{*}=a\hat{Y}^{*}$. Therefore by Theorem \ref{thm:main},
we conclude that $C_{W}(X,Y;D_{1},D_{2})=R{}_{X,Y}(D_{1},D_{2})$
for $\bar{D}_{1}\bar{D}_{2}\leq\rho^{2}$. We note that, if either
$D_{1}$ or $D_{2}$ is greater than $1$, then no information needs
to be sent from the encoder to either one of the two decoders, i.e.:
\begin{equation}
R_{X,Y}(D_{1},D_{2})=\begin{cases}
R_{X}(D_{1}) & \mbox{if }D_{1}<1\,\,,\,\, D_{2}>1\\
R_{Y}(D_{2}) & \mbox{if }D_{1}>1\,\,,\,\, D_{2}<1
\end{cases}
\end{equation}
In the Gray-Wyner network, these bits can be sent only on the respective
private branches, and hence the minimum rate on the shared branch
can be made $0$ while achieving a sum rate of $R_{X,Y}(D_{1},D_{2})$.
Therefore, if $D_{1}$ or $D_{2}$ is greater than $1$, $C_{W}(X,Y;D_{1},D_{2})=0$. 

Next, let us consider the third regime of distortions wherein $\max\{D_{1},D_{2}\}>1-\rho$
and $\bar{D}_{1}\bar{D}_{2}>\rho^{2}$. This corresponds to the two
intermediate regions in Fig. \ref{fig:curve}(a). It is important
to note that the Shannon lower bound is actually tight in this regime.
Therefore, we have $C_{W}(X,Y;D_{1},D_{2})\geq C_{W}(X,Y)$. First,
we show that $C_{W}(X,Y;D_{1},D_{2})>C_{W}(X,Y)$, i.e., it is impossible
to find a joint density over $(X,Y,\hat{X}^{*},\hat{Y}^{*},U^{*})$
satisfying (\ref{eq:mt_2}) and (\ref{eq:mt_3}), where $(X,Y,U^{*})$
satisfy (\ref{eq:Gauss_Lossless_U}) and $(\hat{X}^{*},\hat{Y}^{*})$
achieve joint RD optimality at $(D_{1},D_{2})$. Towards proving this,
note that if the joint density over $(X,Y,U^{*})$ satisfies (\ref{eq:Gauss_Lossless_U}),
then $I(X;U^{*})=I(Y;U^{*})=\frac{1}{2}\log\frac{1}{1-\rho}$. Also
note that any joint density that satisfies (\ref{eq:mt_2}) and (\ref{eq:mt_3})
also satisfies $X\leftrightarrow\hat{X}^{*}\leftrightarrow U^{*}$
and $Y\leftrightarrow\hat{Y}^{*}\leftrightarrow U^{*}$, and therefore
$I(X;\hat{X}^{*})\geq I(X;U^{*})$ and $I(Y;\hat{Y}^{*})\geq I(Y;U^{*})$.
However, if $\max\{D_{1},D_{2}\}>1-\rho$, then $\min\{I(X;\hat{X}^{*}),I(Y;\hat{Y}^{*})\}<\frac{1}{2}\log\frac{1}{1-\rho}$.
Hence, it follows that there exists no such joint density over $(X,Y,\hat{X}^{*},\hat{Y}^{*},U^{*})$,
proving that $C_{W}(X,Y;D_{1},D_{2})>C_{W}(X,Y)$ if $\max\{D_{1},D_{2}\}>1-\rho$
and $(1-D_{1})(1-D_{2})>\rho^{2}$.

Finally, we prove that the lossy CI in this regime of distortions
is given by (\ref{eq:Gauss_Lossy_CI}). Recall that the RD-optimal
random encoder is such that $P(X,Y|\hat{X}^{*},\hat{Y}^{*})=P(X|\hat{X}^{*})P(Y|\hat{Y}^{*})$
and the covariance matrix for $(\hat{X}^{*},\hat{Y}^{*})$ is given
by (\ref{eq:Xs_Ys_Cov_Mat}). Our objective is to find the joint density
over $(X,Y,\hat{X}^{*},\hat{Y}^{*},U)$ satisfying (\ref{eq:mt_2})
and (\ref{eq:mt_3}), that minimizes $I(X,Y;U)$. Clearly, this joint
density additionally satisfies all the following Markov conditions:
\begin{eqnarray}
X\leftrightarrow & \hat{X}^{*} & \leftrightarrow U\nonumber \\
Y\leftrightarrow & \hat{Y}^{*} & \leftrightarrow U\nonumber \\
X\leftrightarrow & U & \leftrightarrow Y\label{eq:add_Mark_cond}
\end{eqnarray}

Hereafter, we restrict ourselves to the regime of distortion where
$D_{1}>1-\rho$ and $\bar{D}_{1}\bar{D}_{2}>\rho^{2}$. This corresponds
to region III in Fig. \ref{fig:curve}(a). Similar arguments hold
for regime IV. Let us consider two extreme possibilities for the random
variable $U$. Observe that both choices $U=\hat{X}^{*}$ and $U=\hat{Y}^{*}$
satisfy all the required Markov conditions. In fact, it is easy to
verify that evaluating $I(X,Y;\hat{X}^{*})$ leads to the expression
for $C_{W}(X,Y;D_{1},D_{2})$ in (\ref{eq:Gauss_Lossy_CI}) for regime
III (and correspondingly evaluating $I(X,Y;\hat{Y}^{*})$ leads to
the lossy CI in regime IV). Hence, we need to prove that, in regime
III, the optimum $U=\hat{X}^{*}$. 

First, we rewrite the objective function as follows:
\begin{eqnarray}
\inf I(X,Y;U) & = & \inf\,\,\, H(X,Y)-H(X,Y|U)\label{eq:alt_obj}\\
 & = & \inf\,\,\, H(X,Y)-H(X|U)-H(Y|U)\nonumber 
\end{eqnarray}
Hence, our objective is equivalent to maximizing $H(X|U)+H(Y|U)$
subject to (\ref{eq:mt_2}), (\ref{eq:mt_3}) and (\ref{eq:add_Mark_cond}).
We will next prove that $U=\hat{X}^{*}$ is the solution to this problem. 

Recall an important result called the data-processing inequality for
minimum mean squared error (MMSE) estimation \cite{verdu_MMSE,Zamir_MMSE}.
It states that, if $X\leftrightarrow Y\leftrightarrow Z$ form a Markov
chain, then $\Phi(X|Y)\leq\Phi(X|Z)$, where $\Phi(A|B)=E\left[\left(A-E\left[A|B\right]\right)^{2}\right]$
is the MMSE of estimating $A$ from $B$. Hence, it follows that for
any joint density $P(X,Y,\hat{X}^{*},\hat{Y}^{*},U)$, satisfying
(\ref{eq:add_Mark_cond}), we have:
\begin{eqnarray}
\Phi(X|U) & \geq & D_{1}\nonumber \\
\Phi(Y|U) & \geq & D_{2}
\end{eqnarray}

Next, we consider a less constrained optimization problem and prove
that the solution to this less constrained problem is bounded by (\ref{eq:Gauss_Lossy_CI}).
It then follows that the solution to (\ref{eq:alt_obj}) is $U=\hat{X}^{*}$.
Consider the following problem:
\begin{equation}
\sup\left\{ H(\tilde{X}|\tilde{U})+H(\tilde{Y}|\tilde{U})\right\} \label{eq:new_form_obj}
\end{equation}
where is supremum is over all joint densities $P(\tilde{X},\tilde{Y},\tilde{U})$, 
subject to the following conditions:
\begin{eqnarray}
(\tilde{X},\tilde{Y}) & \sim & (X,Y)\nonumber \\
\tilde{X}\leftrightarrow & \tilde{U} & \leftrightarrow\tilde{Y}\nonumber \\
\Phi(\tilde{X}|\tilde{U}) & \geq & D_{1}\nonumber \\
\Phi(\tilde{Y}|\tilde{U}) & \geq & D_{2}\label{eq:New_const}
\end{eqnarray}
Observe that all the conditions involving $(\hat{X}^{*},\hat{Y}^{*})$
have been dropped in the above formulation and the constraints for
this problem are a subset of those in (\ref{eq:alt_obj}). We will
next show that the optimum for the above less constrained problem
leads to the expressions in (\ref{eq:Gauss_Lossy_CI}). 

Before proceeding, we show that for any three random variables satisfying
$\tilde{X}\leftrightarrow\tilde{U}\leftrightarrow\tilde{Y}$, where
$\tilde{X}$ and $\tilde{Y}$ are zero mean and of unit variance,
we have: 
\begin{eqnarray}
(1-\Phi(\tilde{X}|\tilde{U}))(1-\Phi(\tilde{Y}|\tilde{U})) & \geq & (E(\tilde{X}\tilde{Y}))^{2}\nonumber \\
 & = & \rho^{2}\label{eq:MMSE_lem}
\end{eqnarray}
Denote the optimal estimators $\theta_{\tilde{X}}(\tilde{U})=E(\tilde{X}|\tilde{U})$
and $\theta_{\tilde{Y}}(\tilde{U})=E(\tilde{Y}|\tilde{U})$. Then,
we have:
\begin{eqnarray}
\Phi(\tilde{X}|\tilde{U}) & = & E\left[\left(\tilde{X}-\theta_{\tilde{X}}(\tilde{U})\right)^{2}\right]\nonumber \\
 & = & E\left[\tilde{X}^{2}\right]-E\left[\left(\theta_{\tilde{X}}(\tilde{U})\right)^{2}\right]\nonumber \\
 & = & 1-E\left[\left(\theta_{\tilde{X}}(\tilde{U})\right)^{2}\right]\label{eq:}
\end{eqnarray}
Therefore, we have:
\begin{eqnarray}
 &  & (1-\Phi(\tilde{X}|\tilde{U}))(1-\Phi(\tilde{Y}|\tilde{U}))\nonumber \\
 &  & =E\left[\left(\theta_{\tilde{X}}(\tilde{U})\right)^{2}\right]E\left[\left(\theta_{\tilde{Y}}(\tilde{U})\right)^{2}\right]\nonumber \\
 &  & \geq^{(a)}\left(E\left[\theta_{\tilde{Y}}(\tilde{U})\theta_{\tilde{Y}}(\tilde{U})\right]\right)^{2}\nonumber \\
 &  & =^{(b)}\left(E\left[E\left[\tilde{X}\tilde{Y}\Bigl|\tilde{U}\right]\right]\right)^{2}\nonumber \\
 &  & =\left(E\left[\tilde{X}\tilde{Y}\right]\right)^{2}=\rho^{2}\label{eq:Series_2}
\end{eqnarray}
where (a) follows from Cauchy-Schwarz inequality and (b) follows from
the Markov condition $\tilde{X}\leftrightarrow\tilde{U}\leftrightarrow\tilde{Y}$. 

This allows us to further simplify the formulation in (\ref{eq:new_form_obj}).
Specifically, we relax the constraints (\ref{eq:New_const}) by imposing
(\ref{eq:MMSE_lem}), instead of the Markov condition $\tilde{X}\leftrightarrow\tilde{U}\leftrightarrow\tilde{Y}$.
Our objective now becomes:

\begin{equation}
\sup\left\{ H(\tilde{X}|\tilde{U})+H(\tilde{Y}|\tilde{U})\right\} \label{eq:new_form_obj-1}
\end{equation}
where the supremum is over all joint densities $P(\tilde{X},\tilde{Y},U)$,
subject to the following conditions:
\begin{eqnarray}
(\tilde{X},\tilde{Y}) & \sim & (X,Y)\nonumber \\
(1-\Phi(\tilde{X}|\tilde{U}))(1-\Phi(\tilde{Y}|\tilde{U})) & \geq & \rho^{2}\nonumber \\
\Phi(\tilde{X}|\tilde{U}) & \geq & D_{1}\nonumber \\
\Phi(\tilde{Y}|\tilde{U}) & \geq & D_{2}\label{eq:New_const-1}
\end{eqnarray}
We next bound $H(\tilde{X}|\tilde{U})+H(\tilde{Y}|\tilde{U})$ in
terms of the corresponding $MMSE$ as:
\begin{eqnarray}
H(\tilde{X}|\tilde{U})+H(\tilde{Y}|\tilde{U}) & \leq & \frac{1}{2}\log\left(2\pi e\Phi(\tilde{X}|\tilde{U})\right)\nonumber \\
 &  & +\frac{1}{2}\log\left(2\pi e\Phi(\tilde{Y}|\tilde{U})\right)\label{eq:bound_H_MMSE}
\end{eqnarray}
Using (\ref{eq:bound_H_MMSE}) to bound (\ref{eq:new_form_obj-1})
leads to the following objective function:
\begin{equation}
\sup\left\{ \Phi(\tilde{X}|\tilde{U})\Phi(\tilde{Y}|\tilde{U})\right\} 
\end{equation}
subject to the conditions in (\ref{eq:New_const-1}). It is easy to
verify that the maximum for this objective function satisfying (\ref{eq:New_const-1})
is achieved either at $(\Phi(\tilde{X}|\tilde{U})=D_{1},\Phi(\tilde{Y}|\tilde{U})=1-\frac{\rho^{2}}{1-D_{1}})$
or at $(\Phi(\tilde{X}|\tilde{U})=1-\frac{\rho^{2}}{1-D_{2}},\Phi(\tilde{Y}|\tilde{U})=D_{2})$,
depending on whether $D_{1}>1-\rho$ or $D_{2}>1-\rho$. Substituting
these values in (\ref{eq:bound_H_MMSE}) leads to upper bounds on
$H(\tilde{X}|\tilde{U})+H(\tilde{Y}|\tilde{U})$ for the two distortion
regimes, respectively. These upper bounds are achieved by setting
$U=\hat{X}^{*}$ or $U=\hat{Y}^{*}$, depending on the distortion
regime. The proof of the theorem follows by noting that these choices
for $U$ lead to the lossy CI being equal to (\ref{eq:Gauss_Lossy_CI}).
Therefore, we have completely characterized $C_{W}(X,Y;D_{1},D_{2})$
for $(X,Y)$ jointly Gaussian for all distortions $(D_{1},D_{2})>0$. 
\end{proof}

\section{Lossy Extension of the Gács-Körner Common Information\label{sec:Gacs-Korner's-CI}}

\subsection{Definition}

Recall the definition of the Gács-Körner CI from Section \ref{sec:Prior-Results}.
Although the original definition does not have a direct lossy interpretation,
the equivalent definition given by Ahlswede and Körner, in terms of
the lossless Gray-Wyner region can be extended to the lossy setting,
similar to our approach to Wyner's CI. These generalizations provide
theoretical insight into the performance limits of practical databases
for fusion storage of correlated sources as described in \cite{ITW_CI}.

We define the lossy generalization of the Gács-Körner CI at $(D_{1},D_{2})$,
denoted by $C_{GK}(X,Y;D_{1},D_{2})$ as follows. Let $\mathcal{R}_{GK}(D_{1},D_{2})$
be the set of $R_{0}$ such that for any $\epsilon>0$, there exists
a point $(R_{0},R_{1},R_{2})$ satisfying the following conditions:
\begin{eqnarray}
(R_{0},R_{1},R_{2})\in\mathcal{R}_{GW}(D_{1},D_{2})
\end{eqnarray}
\begin{eqnarray}
R_{0}+R_{1}\leq R{}_{X}(D_{1})+\epsilon &  & R_{0}+R_{2}\leq R{}_{Y}(D_{2})+\epsilon
\end{eqnarray}
Then, 
\begin{equation}
C_{GK}(X,Y;D_{1},D_{2})=\sup\,\, R_{0}\in\mathcal{R}_{GK}(D_{1},D_{2})
\end{equation}
Again, observe that, if every point on the intersection of the planes
$R_{0}+R_{1}=R_{X}(D_{1})$ and $R_{0}+R_{2}=R_{X}(D_{2})$ satisfies
(\ref{eq:GW_Lossy_Region}) with equality for some joint density $P(X,Y,\hat{X},\hat{Y},U)$,
then the above definition can be simplified by setting $\epsilon=0$.
Hereafter, we will assume that this condition holds, noting that the
results can be extended to the general case, similar to arguments
in Appendix A.

\subsection{Single Letter Characterization of $C_{GK}(X,Y;D_{1},D_{2})$}

We provide an information theoretic characterization for $C_{GK}(X,Y;D_{1},D_{2})$
in the following theorem. 
\begin{thm}
\label{thm:Thm_GK}A single letter characterization of $C_{GK}(X,Y;D_{1},D_{2})$
is given by:
\begin{equation}
C_{GK}(X,Y;D_{1},D_{2})=\sup I(X,Y;U)
\end{equation}
where the supremum is over all joint densities $(X,Y,\hat{X},\hat{Y},U)$
such that the following Markov conditions hold:
\begin{eqnarray}
Y\leftrightarrow X\leftrightarrow U &  & X\leftrightarrow Y\leftrightarrow U\nonumber \\
X\leftrightarrow\hat{X}\leftrightarrow U &  & Y\leftrightarrow\hat{Y}\leftrightarrow U\label{eq:GK_Thm_mc}
\end{eqnarray}
where $P(\hat{X}|X)\in\mathcal{P}_{D_{1}}^{X}$ and $P(\hat{Y}|Y)\in P_{D_{2}}^{Y}$,
are any rate-distortion optimal encoders at $D_{1}$ and $D_{2}$,
respectively. \end{thm}
\begin{proof}
The proof follows in very similar lines to the proof of Theorem \ref{thm:main}.
The original Gray-Wyner characterization is, in fact, sufficient in
this case. We first assume that there are unique encoders $P(\hat{X}|X)$
and $P(\hat{Y}|Y)$, that achieve $R{}_{X}(D_{1})$ and $R{}_{Y}(D_{2})$,
respectively. The proof extends directly to the case of multiple rate-distortion
optimal encoders. 

We are interested in characterizing the points in $\mathcal{R}_{GW}(D_{1},D_{2})$
which lie on both the planes $R_{0}+R_{1}=R{}_{X}(D_{1})$ and $R_{0}+R_{2}=R{}_{Y}(D_{2})$.
Therefore we have the following series of inequalities:
\begin{eqnarray}
R{}_{X}(D_{1}) & = & R_{0}+R_{1}\nonumber \\
 & \geq & I(X,Y;U)+I(X;\hat{X}|U)\nonumber \\
 & = & I(X;\hat{X},U)+I(Y;U|X)\nonumber \\
 & \geq & I(X;\hat{X})\geq R{}_{X}(D_{1})
\end{eqnarray}
Writing similar inequality relations for $Y$ and following the same
arguments as in Theorem \ref{thm:main}, it follows that for all joint
densities satisfying (\ref{eq:GK_Thm_mc}) and for which $P(\hat{X}|X)\in\mathcal{P}_{D_{1}}^{X}$
and $P(\hat{Y}|Y)\in\mathcal{P}_{D_{2}}^{Y}$, there exists at least
one point in $\mathcal{R}_{GW}(D_{1},D_{2})$ which satisfies both
$R_{0}+R_{1}=R{}_{X}(D_{1})$ and $R_{0}+R_{2}=R{}_{Y}(D_{2})$ and
for which $R_{0}=I(X,Y;U)$. This proves the theorem.\end{proof}
\begin{cor}
\label{cor:CGK}$C_{GK}(X,Y;D_{1},D_{2})\leq C_{GK}(X,Y)$\end{cor}
\begin{proof}
This corollary follows directly from Theorem \ref{thm:Thm_GK} as
conditions in (\ref{eq:GK_AK_mc}) are a subset of the conditions
in (\ref{eq:GK_Thm_mc}).
\end{proof}
It is easy to show that if the random variables $(X,Y)$ are jointly
Gaussian with a correlation coefficient $\rho<1$, then $C_{GK}(X,Y)=0$.
Hence from Corollary \ref{cor:CGK}, it follows that, for jointly
Gaussian random variables with correlation coefficient strictly less
than $1$, $C_{GK}(X,Y;D_{1},D_{2})=0\,\,\forall D_{1},D_{2}$, under
any distortion measure. It is well known that $C_{GK}(X,Y)$ is typically
very small and is non-zero only when the ergodic decomposition of
the joint distribution leads to non-trivial subsets of the alphabet
space. In the general setting, as $C_{GK}(X,Y;D_{1},D_{2})\leq C_{GK}(X,Y)$,
it would seem that Theorem \ref{thm:Thm_GK} has very limited practical
significance. However, in \cite{CGK_SR} we showed that $C_{GK}(X,Y;D_{1},D_{2})$
plays a central role in scalable coding of sources that are not successively
refinable. Further implications of this result will be studied in
more detail as part of future work.

\section{Optimal Transmit-Receive Rate Tradeoff in Gray-Wyner Network \label{sec:trade-off}}

\subsection{Motivation\label{sec:Motivation}}

It is well known that the two definitions of CI, due to Wyner and
Gács-Körner, can be characterized using the Gray-Wyner region and
the corresponding operating points are two boundary points of the
region. Several approaches have been proposed to provide further insight
into the underlying connections between them \cite{Ahlswede74oncommon,dual_CI,effros_CI,Yamamoto_CI}.
However, to the best of our knowledge, no prior work has identified
an operationally significant contour of points on the Gray-Wyner region,
which connects the two operating points. In this section we derive
and analyze such a contour of points on the boundary, obtained by
trading-off a particular definition of transmit rate and receive rate,
which passes through both the operating points of Wyner and Gács-Körner.
This tradeoff provides a generic framework to understand the underlying
principles of shared information. We note in passing that Paul Cuff
characterized a tradeoff between Wyner's common information and the
mutual information in \cite{Paul_Cuff}, while studying the amount
of common randomness required to generate correlated random variables,
with communication constraints. This tradeoff is similar in spirit
to the tradeoff studied in this paper, although very different in
details. 

We define the total transmit rate for the Gray-Wyner network as $R_{t}=R_{0}+R_{1}+R_{2}$
and the total receive rate as $R_{r}=2R_{0}+R_{1}+R_{2}$. Specifically,
we show that the contour traced on the Gray-Wyner region boundary
when we trade $R_{t}=R_{0}+R_{1}+R_{2}$ for $R_{r}=2R_{0}+R_{1}+R_{2}$,
passes through both the operating points of Wyner and Gács-Körner
for any distortion pair $(D_{1},D_{2})$.

\subsection{Relating the Two notions of CI\label{sec:Relation-to-Common}}

\begin{figure}[!t]
\centering\includegraphics[scale=0.4]{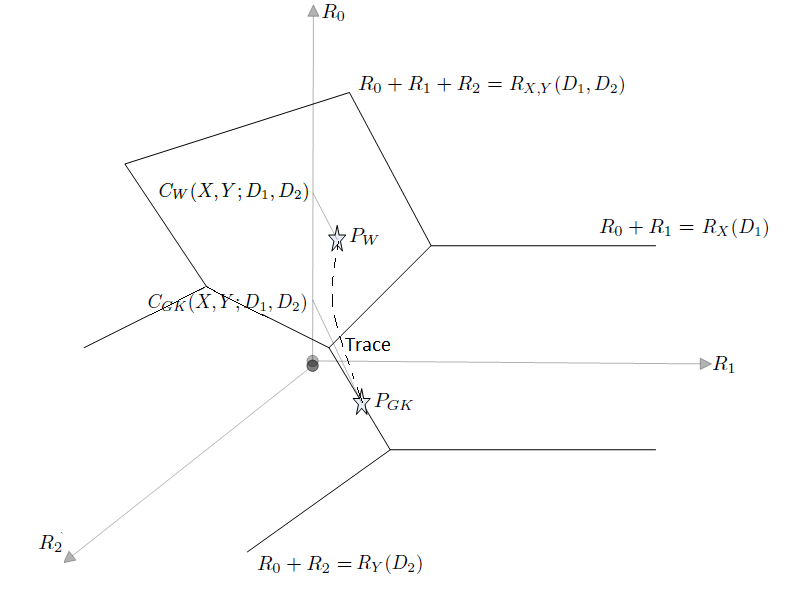}\caption{\label{fig:GW_region}$P_{W}$ and $P_{GK}$ in the Gray-Wyner region.
Observe that the transmit contour and the receive contour coincide
in between $P_{W}$ and $P_{GW}$.}
\end{figure}

Let $P_{W}$ and $P_{GK}$ denote the respective operating points
in $\mathcal{R}_{GW}(D_{1},D_{2})$, corresponding to the lossy definitions
of Wyner and Gács-Körner CI ($C_{W}(X,Y;D_{1},D_{2})$ and $C_{GK}(X,Y;D_{1},D_{2})$).
Hereafter, the dependence on the distortion constraints will be implicit
for notational convenience. $P_{W}$ and $P_{GK}$ are shown in Fig.
\ref{fig:GW_region}. 

We define the following two contours in the Gray-Wyner region. We
define the transmit contour as the set of points on the boundary of
the Gray-Wyner region obtained by minimizing the total receive rate
$(R_{r})$ at different total transmit rates $(R_{t})$, i.e., the
transmit contour is the trace of operating points obtained when the
receive rate is minimized subject to a constraint on the transmit
rate. Similarly, we define the receive contour as the trace of points
on the Gray-Wyner region obtained by minimizing the transmit rate
$(R_{t})$ for each $(R_{r})$.

\textbf{Claim}: The transmit contour coincides with the receive contour
between $P_{W}$ and $P_{GK}$. 
\begin{proof}
$\mathcal{R}_{GW}$ is a convex region. Hence the set of achievable
rate pairs for $(R_{t},R_{r})=(R_{0}+R_{1}+R_{2},2R_{0}+R_{1}+R_{2})$
is convex. We have $R_{t}\geq R_{X,Y}(D_{1},D_{2})$ and $R_{r}\geq R_{X}(D_{1})+R_{Y}(D_{2})$.
Note that when $R_{t}=R_{X,Y}(D_{1},D_{2})$, $\min R_{r}=R_{X,Y}(D_{1},D_{2})+C_{W}(X,Y;D_{1},D_{2})$
and is achieved at $P_{W}$. Similarly when $R_{r}=R_{X}(D_{1})+R_{Y}(D_{2})$,
$\min R_{t}=R_{X}(D_{1})+R_{Y}(D_{2})-C_{GK}(X,Y;D_{1},D_{2})$, which
is achieved at $P_{GK}$. Figure \ref{fig:Rt_Rr_Tradeoff} depicts
the trade-off between $R_{t}$ and $R_{r}$. Hence, it follows from
the convexity of $(R_{t},R_{r})$ region that for every transmit rate
$R_{t}\geq R_{X,Y}(D_{1},D_{2})$, there exists a receive rate $R_{r}\leq R_{X}(D_{1})+R_{Y}(D_{2})$,
such that, the corresponding operating points on the transmit and
the receive contours respectively coincide. Hence, it follows that
the transmit and the receive contours coincide in between $P_{W}$
and $P_{GK}$. 
\end{proof}
This new relation between the two notions of CI brings them both under
a common framework. Gács and Körner's operating point can now be stated
as the \textit{minimum} shared rate (similar to Wyner's definition),
at a sufficiently large sum rate. Likewise, Wyner's CI can be defined
as the \textit{maximum} shared rate (similar to Gács and Körner's
definition), at a sufficiently large receive rate. We will make these
arguments more precise in Section \ref{sub:Relation-to-Gacs-Korner}
when we derive alternate characterizations in the lossless framework
for each notion in terms of the objective function of the other.

We note that operating point corresponding to Gács-Körner CI is always
unique, for all $(D_{1},D_{2})$, as it lies on the intersection of
two planes. However, the operating point corresponding to Wyner's
CI may not be unique, i.e., there could exist source and distortion
pairs for which, the minimum shared rate on the Pangloss plane could
be achieved at multiple points in the GW region. In fact, it follows
directly from the convexity of the GW region that if there are two
operating points corresponding to Wyner's CI, then all points in between
them also achieve minimum shared rate and lie on the Pangloss plane.
For such sources and distortion pairs, the trade-off between the transmit
and receive rates on the GW network leads to a surface of operating
points, instead of a contour. Nevertheless, this surface always intersects
both $P_{W}$ and $P_{GK}$. 

\begin{figure}[!t]
\centering \includegraphics[scale=0.34]{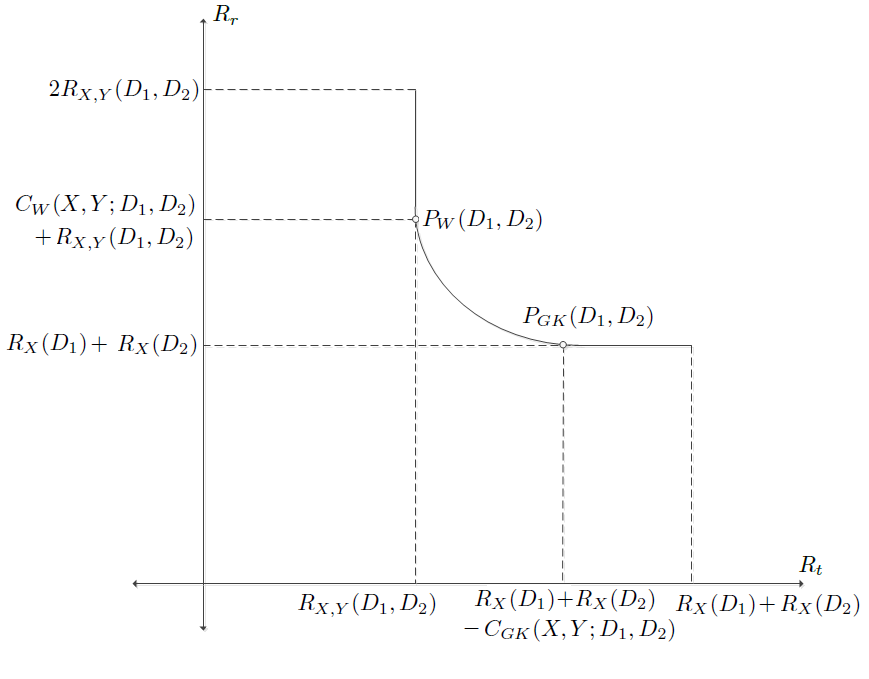}\caption{Tradeoff between $R_{r}$ and $R_{t}$. Observe that, when $R_{t}=R_{X,Y}(D_{1},D_{2})$,
minimum $R_{r}$ is equal to $R_{X,Y}(D_{1},D_{2})+C_{W}(X,Y;D_{1},D_{2})$
and when $R_{r}=R_{X}(D_{1})+R_{Y}(D_{2})$, minimum $R_{t}=R_{X}(D_{1})+R_{Y}(D_{2})-C_{GK}(X,Y;D_{1},D_{2})$\label{fig:Rt_Rr_Tradeoff}}
\end{figure}

\subsection{Single Letter Characterization of the Tradeoff\label{sec:Formal-statement-and}}

The tradeoff between the transmit and the receive rates in the Gray-Wyner
network not only plays a crucial role in providing theoretical insight
into the workings of the two notions of CI, but also has implications
in several practical scenarios such as fusion coding and selective
retrieval of correlated sources in a database \cite{fusion} and dispersive
information routing of correlated sources \cite{DIR}, as described
in \cite{ITW_CI}. It is therefore of interest to derive a single
letter information theoretic characterization for this tradeoff. Although
we were unable to derive a complete characterization for general distortions,
we derive a single letter complete characterization for the lossless
setting here. Hence, we focus only on the lossless setting for the
rest of the paper. 

To gain insight into this tradeoff, we characterize two curves which
are rotated/transformed versions of each other. The first curve, denoted
by $C(X,Y;R^{'})$, plots the minimum shared rate, $R_{0}$, at a
transmit rate of $H(X,Y)+R^{'}$ and the second, denoted by $K(X,Y;R^{''})$,
is the maximum $R_{0}$ at a receive rate of $H(X)+H(Y)+R^{''}$.
It is easy to see that the transmit-receive rate tradeoff can be derived
directly from these quantities. We note that the quantities $C(X,Y;R^{'})$
and $K(X,Y;R^{''})$ are in fact generalizations of Wyner and Gács-Körner
(lossless) definitions of CI to the excess sum transmit rate and receive
rate regimes, respectively. Using their properties, we will also derive
alternate characterizations for the two notions of lossless CI under
a unified framework in Section \ref{sec:Properties-of-the}.

We define the quantity $C(X,Y;R^{'})\,\,\forall R^{'}\in[0,I(X,Y)]$
as:
\begin{eqnarray}
C(X,Y;R^{'}) & = & \inf\, R_{0}:(R_{0},R_{1},R_{2})\in\mathcal{R}_{GW}\label{eq:defn_SI}
\end{eqnarray}
 satisfying,
\begin{equation}
R_{0}+R_{1}+R_{2}=H(X,Y)+R^{'}\label{eq:eq_thm_1}
\end{equation}
 Similarly, we define the quantity $K(X,Y;R^{''})\,\,\forall R^{''}\in[0,H(X,Y)-I(X,Y)]$
as: 
\begin{eqnarray}
K(X,Y;R^{''}) & = & \sup\, R_{0}:(R_{0},R_{1},R_{2})\in\mathcal{R}_{GW}\label{eq:defn_SI-1}
\end{eqnarray}
 satisfying,
\begin{equation}
2R_{0}+R_{1}+R_{2}=H(X)+H(Y)+R^{''}\label{eq:eq_thm_1-1}
\end{equation}
 Note that we restrict the ranges for $R^{'}$ and $R^{''}$ to the
ranges of practical interest, as operating at $R^{'}>I(X;Y)$ or $R^{''}>H(X,Y)-I(X,Y)$
is suboptimal and uninteresting. The following Theorem provides information
theoretic characterizations for $C(X,Y;R^{'})$ and $K(X,Y;R^{''})$. 
\begin{thm}
\label{thm:(i)excess_rate}(i) For any excess sum transmit rate $R^{'}\in[0,I(X,Y)]$:
\begin{equation}
C(X,Y;R^{'})=\min\, I(X,Y;U)\label{eq:eq_1_thm_1}
\end{equation}
 where the minimization is over all $U$ jointly distributed with
$(X,Y)$ such that:
\begin{equation}
I(X;Y|U)=R^{'}\label{eq:eq_2_thm_1}
\end{equation}
 We denote the operating point in $\mathcal{R}_{GW}$ corresponding
to the minimum by $P_{C(X,Y)}(R^{'})$. \\
 (ii) For any excess reception rate $R^{''}\in[0,H(X,Y)-I(X,Y)]$:
\begin{equation}
K(X,Y;R^{''})=\max\, I(X,Y;W)\label{eq:eq_1_thm_1-1}
\end{equation}
 where the maximization is over all $W$ jointly distributed with
$(X,Y)$ such that:
\begin{equation}
I(X;W|Y)+I(Y;W|X)=R^{''}\label{eq:eq_2_thm_1-1}
\end{equation}
 We denote the operating point in $\mathcal{R}_{GW}$ corresponding
to the maximum by $P_{K(X,Y)}(R^{''})$.\end{thm}
\begin{proof}
We prove part (i) of the theorem for $C(X,Y;R^{'})$. The proof of
(ii) for $K(X,Y;R^{''})$ follows similar lines.

\textbf{\textit{Achievability}} : Let $U$ be jointly distributed
with $(X,Y)$ such that $I(X;Y|U)=R^{'}$. It leads to a point in
the Gray-Wyner region with $(R_{0},R_{1},R_{2})=(I(X,Y;U),$ $H(X|U),H(Y|U))$.
On substituting in (\ref{eq:eq_thm_1}) we have:
\begin{eqnarray}
R_{0}+R_{1}+R_{2} & = & I(X,Y;U)+H(X|U)+H(Y|U)\nonumber \\
 & = & H(X,Y)+I(X;Y|U)\\
 & = & H(X,Y)+R^{'}
\end{eqnarray}
Note that the existence of a $U$ that achieves the minimum in (\ref{eq:eq_1_thm_1})
follows from Theorem 4.4 (A) in \cite{Wyner_CI}. This allows us to replace the infimum in the definition of $C(X,Y;R^{'})$ with a minimum in (\ref{eq:eq_1_thm_1}).

\textbf{\textit{Converse}}\textbf{ }: We know from the converse to
the Gray-Wyner region that every point in $\mathcal{R}_{GW}$ is achieved
by some random variable $U$ jointly distributed with $(X,Y)$. We
need to determine the condition on $U$ for (\ref{eq:eq_thm_1}) to
hold. On substituting $(R_{0},R_{1},R_{2})=(I(X,Y;U),$ $H(X|U),H(Y|U))$
in (\ref{eq:eq_thm_1}), we get the condition to be (\ref{eq:eq_2_thm_1}),
proving the converse. 
\end{proof}
Note that the cardinality of $\mathcal{U}$ can be restricted to $|\mathcal{U}|\leq$
$|\mathcal{X}|$$|\mathcal{Y}|$ + 1 using Theorem 4.4 in \cite{Wyner_CI}.
Also note that when the transmit rate is $H(X,Y)+R^{'}$, the minimum
receive rate is $H(X,Y)+R^{'}+C(X,Y;R^{'})$. Similarly, when the
receive rate is $H(X)+H(Y)+R^{''}$, the minimum transmit rate is
$H(X)+H(Y)+R^{''}-K(X,Y;R^{''})$. Hence the quantities $C(X,Y;R^{'})$
and $K(X,Y;R^{''})$ are just rotated/transformed versions of the
transmit versus receive rate tradeoff curve.

We refer to plots of $C(X,Y;R^{'})$ and $K(X,Y;R^{''})$ versus $R^{'}$
and $R^{''}$ as the \textit{`transmit tradeoff curve}' and the \textit{`receive
tradeoff curve}', respectively. Observe that in both cases, as we
increase $R^{'}$ (or $R^{''}$) we obtain parallel cross-sections
of the Gray-Wyner region and the set of operating points $P_{C(X,Y)}(R^{'})$
and $P_{K(X,Y)}(R^{''})$, trace contours on the boundary of the region.
These two contours are precisely the transmit and the receive contours
defined in section \ref{sec:Relation-to-Common}. Note the difference
between the contours and their respective tradeoff curves. The contours
are defined in a 3-D space and lie on the boundary of $\mathcal{R}_{GW}$.
In general each of the contours may not even lie on a single plane.
However, the tradeoff curves are a projection of the respective contours
on to a 2-D plane.

\subsection{Properties of the tradeoff curve\label{sec:Properties-of-the}}

\begin{figure}[!t]
\centering \includegraphics[scale=0.28]{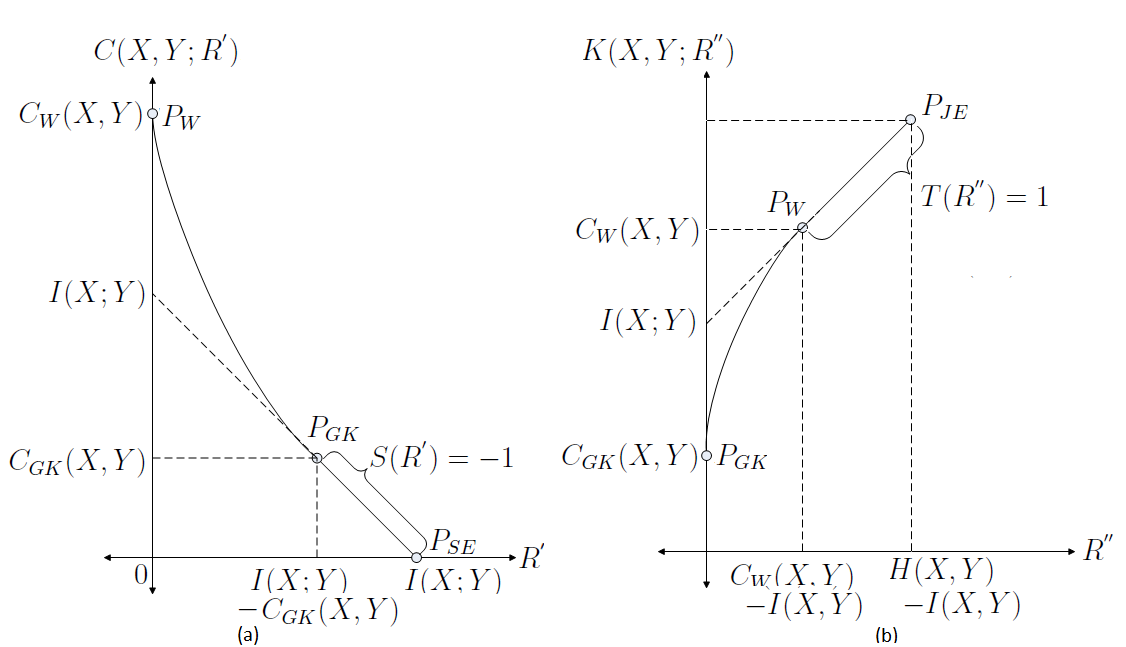}\caption{(a) Typical transmit tradeoff curve - $C(X,Y;R^{'})$ (b) Typical
receive tradeoff curve - $K(X,Y;R^{''})$\label{fig:Typical-Common-Information}}
\end{figure}

In this section, we focus on the quantities $C(X,Y;R^{'})$ and $K(X,Y;R^{''})$
and analyze some important properties, which allow us to provide alternate
characterizations for the two notions of CI. As most of the proofs
for $K(X,Y;R^{''})$ are very similar to their $C(X,Y;R^{'})$ counterparts,
we only prove the properties for $C(X,Y;R^{'})$. We plot typical
transmit and receive tradeoff curves in Figure \ref{fig:Typical-Common-Information}
to illustrate the discussion.

Consider the transmit tradeoff curve. At $R^{'}=0$ we get the operating
point corresponding to Wyner CI where the minimum shared information
is given by $C_{W}(X,Y)$. This point is denoted by $P_{W}$ in Figure
\ref{fig:Typical-Common-Information}. Next observe that at $R_{0}=0,$
for lossless reconstruction of $X$ and $Y$, we need, $R_{1}\geq H(X)$
and $R_{2}\geq H(Y)$. Therefore at an excess sum transmit rate $R^{'}=H(X)+H(Y)-H(X,Y)=I(X;Y)$,
the shared rate vanishes, or, $C(X,Y;I(X;Y))=0$. We call this point
- `separate encoding' and denote it by $P_{SE}$ in the figure. It
is also obvious that any $U$ independent of $(X,Y)$ achieves this
minimum $R_{0}$ for $R^{'}=I(X;Y)$. 
\begin{lem}
\textbf{Convexity:} \label{lem:convex}$C(X,Y;R^{'})$ is convex for $R^{'}\in[0,I(X;Y)]$
and $K(X,Y;R^{''})$ is concave for $R^{''}\in[0,H(X,Y)-I(X;Y)]$.\end{lem}
\begin{proof}
The proof follows directly from the convexity of the Gray-Wyner region.\end{proof}
\begin{lem}
\textbf{Monotonicity} :\label{lem:monoton}$C(X,Y;R^{'})$ is strictly
monotone decreasing $\forall R^{'}\in[0,I(X;Y)]$ and $K(X,Y;R^{''})$
is strictly monotone increasing $\forall R^{''}\in[0,H(X,Y)-I(X;Y)]$\end{lem}
\begin{proof}
It is clear from the achievability results of Gray-Wyner that if a
point $(r_{0},r_{1},r_{2})\in\mathcal{R}_{GW}$, then all points $\{(R_{0},R_{1},R_{2}):R_{0}\geq r_{0},R_{1}\geq r_{1},R_{2}\geq r_{2}\}\in\mathcal{R}_{GW}$.
Let $C(X,Y;R^{'})=r_{0}$, for some excess transmission rate $R^{'}$
and let the corresponding operating point in $\mathcal{R}_{GW}$ be
$(r_{0},r_{1},r_{2})$. Hence for any $\Delta>0$, the point $(r_{0},r_{1}+\Delta,r_{2})\in\mathcal{R}_{GW}$
and satisfies $R_{0}+R_{1}+R_{2}=R+\Delta$. Therefore, 
\begin{eqnarray}
C(X,Y;R^{'}+\Delta) & = & \inf R_{0}:\{R_{0}+R_{1}+R_{2}=R+\Delta\}\nonumber \\
 & \leq & r_{0}
\end{eqnarray}
Hence, $C(X,Y;R^{'})$ is non-increasing. Then it follows from convexity
that $C(X,Y;R^{'})$ is either a constant or is strictly monotone
decreasing. Lemma \ref{lem:slope} below eliminates the possibility
of a constant, proving this lemma. 
\end{proof}
At all $R^{'}$ where $C(X,Y;R^{'})$ is differentiable, we denote
the slope by $S(R^{'})$. At non-differentiable points, we denote
by $S^{-}(R^{'})$ and $S^{+}(R^{'})$ the left and right derivatives,
respectively. Similarly the slope, left derivative and right derivatives
of $K(X,Y;R^{''})$ are denoted by $T(R^{''})$, $T^{-}(R^{''})$
and $T^{+}(R^{''})$, respectively. 
\begin{lem}
\label{lem:slope}The slope of $C(X,Y;R^{'})$, $S(R^{'})\leq-1$
$\forall R\in[0,I(X;Y)]$ where the curve is differentiable. At non-differentiable
points, we have $S^{-}(R^{'})<S^{+}(R^{'})\leq-1$. Similarly we have,
$T(R^{''})\geq1\,\,\forall R^{''}\in[0,H(X,Y)-I(X;Y)]$ and $T^{-}(R^{''})>T^{+}(R^{''})\geq1$\end{lem}
\begin{proof}
Note that it is sufficient for us to show that $S^{-}(I(X;Y))\leq-1$.
Then it directly follows from convexity that $S(R^{'})\leq-1$ at
all differentiable points and $S^{-}(R^{'})<S^{+}(R^{'})\leq-1$ at
all non-differentiable points. Consider $\triangle>0$, and fix the
shared information rate to be $R_{0}=\triangle$. From the converse
of the source coding theorem for lossless reconstruction, we have:
\begin{eqnarray}
R_{0}+R_{1}=\triangle+R_{1} & \geq & H(X)\nonumber \\
R_{0}+R_{2}=\triangle+R_{2} & \geq & H(Y)\label{eq:slope}
\end{eqnarray}
 The above inequalities imply $R_{0}+R_{1}+R_{2}\geq H(X,Y)+I(X;Y)-\triangle$.
Therefore the point on the transmit tradeoff curve with $C(X,Y;R^{'})=\triangle$
has $R^{'}\geq I(X;Y)-\triangle$. Hence $S^{-}(I(X;Y))\leq-1$ proving
the Lemma. \end{proof}
\begin{rem}
Observe that the above proofs do not rely on the lossless definitions
of $C(X,Y;R^{'})$ and $K(X,Y;R^{''})$, but leverage only on the
convexity of the lossless Gray-Wyner region. It is well known that
the lossy Gray-Wyner region is also convex in the rates, for all distortion
pairs. Consequently, all the three lemmas (\ref{lem:convex}, \ref{lem:monoton}
and \ref{lem:slope}) can be easily extended to the lossy counterparts
of $C(X,Y;R^{'})$ and $K(X,Y;R^{''})$. We omit the details of the
proof here to avoid repetition. 
\end{rem}

\subsection{Alternate characterizations for $C_{GK}(X,Y)$ and $C_{W}(X,Y)$\label{sub:Relation-to-Gacs-Korner}}

In this section, we provide alternate characterizations for $C_{GK}(X,Y)$
and $C_{W}(X,Y)$ in terms of $C(X,Y;R^{'})$ and $C(X,Y;R^{''})$,
respectively. 
\begin{thm}
\label{thm:GK_C}An alternate characterization for the Gács-Körner
CI is:
\begin{eqnarray}
C_{GK}(X,Y) & = & \sup_{R^{'}:S^{+}(R^{'})=-1}C(X,Y;R^{'})\label{eq:GK_char}
\end{eqnarray}
 If there exists no $R^{'}$ for which $S^{+}(R^{'})=-1$, then, $C_{GK}(X,Y)=0$.
Similarly, an alternate characterization for Wyner's CI is : 
\begin{eqnarray}
C_{W}(X,Y) & = & \inf_{R^{''}:T^{+}(R^{''})=1}K(X,Y;R^{''})\label{eq:GK_char-1}
\end{eqnarray}
 If there exists no $R^{''}$ for which $T^{+}(R)=1$, then, $C_{W}(X,Y)=H(X,Y)$.
Note that $C_{GK}(X,Y)$ corresponds to that excess sum transmit rate
where the region of $C(X,Y;R^{'})$ with slope $<-1$ meets the region
with slope equal to $-1$, and $C_{W}(X,Y)$ corresponds to that excess
receive rate where the region of $K(X,Y;R^{''})$ with slope $>1$
meets the region with slope equal to $1$. \end{thm}
\begin{proof}
We first assume that there exists some $R^{*}\in[0,I(X;Y))$, which
is the minimum rate at which $S^{+}(R^{*})=-1$. We need to show that
$C(X,Y;R^{*})=C_{GK}(X,Y)$. We denote this point by $P_{GK}$ in
the figure. Let $\tilde{R}$ be such that $R^{*}\leq\tilde{R}<I(X;Y)$
and let $\tilde{U}$ be the random variable which achieves the minimum
shared information rate at $\tilde{R}$ in Theorem \ref{thm:(i)excess_rate}.
Then it follows from Lemmas \ref{lem:convex} and \ref{lem:slope}
that $S^{+}(\tilde{R})=-1$. Then the point in the GW region corresponding
to $\tilde{U}$ satisfies the following two conditions:
\begin{eqnarray}
R_{0} & = & I(X,Y)-\tilde{R}\nonumber \\
R_{0}+R_{1}+R_{2} & = & H(X,Y)+\tilde{R}
\end{eqnarray}
 Adding the two equations, we have $2R_{0}+R_{1}+R_{2}=H(X)+H(Y)$,
which implies that $R_{0}+R_{1}=H(X)$ and $R_{0}+R_{2}=H(Y)$. Therefore,
the point corresponding to $\tilde{U}$ satisfies Gács-Körner constraints
(\ref{eq:constraint_AK_GK}). Hence, it follows that, any $\tilde{R}$
such that $S^{+}(\tilde{R})=-1$ leads to an operating point in the
GW region which satisfies Gács-Körner constraints.

Next, we need to show the converse. Consider any point in the GW region
satisfying Gács-Körner constraints. It can be written as, 
\begin{eqnarray}
R_{0} & = & I(X;Y)-\tilde{R}\nonumber \\
R_{1} & = & H(X)-(I(X;Y)-\tilde{R})\nonumber \\
R_{2} & = & H(Y)-(I(X;Y)-\tilde{R})\label{eq:GC_relation_to_C}
\end{eqnarray}
 for some $C_{GK}(X,Y)\leq\tilde{R}\leq I(X;Y)$. On summing the three
equations, we have $R_{0}+R_{1}+R_{2}=H(X,Y)+\tilde{R}$. It then
follows from the convexity of $C(X,Y;R^{'})$ that $S^{+}(\tilde{R})=-1$.
Therefore, we have, 
\begin{eqnarray}
C(X,Y;R^{*}) & = & I(X;Y)-R^{*}\nonumber \\
 & = & I(X;Y)-\min_{\tilde{R}:S^{+}(\tilde{R})=-1}\tilde{R}\nonumber \\
 & = & \max\,\Bigl\{ R_{0}:(R_{0},R_{1},R_{2})\in\mathcal{R}_{GW}\nonumber \\
 &  & R_{0}+R_{1}=H(X),\,\, R_{0}+R_{2}=H(Y)\Bigr\}\nonumber \\
 & = & C_{GK}(X,Y)
\end{eqnarray}
 proving the first part of the theorem. However, if there exists no
$R^{'}\in[0;I(X;Y)]$ for which $S^{+}(R^{'})=-1$, it implies that
$\forall R^{'}\in[0;I(X;Y)),\,\, C(X;Y;R^{'})>I(X;Y)$. Therefore
(\ref{eq:constraint_AK_GK}) is not satisfied with equality for any
$R^{'}\in[0;I(X;Y))$. Hence $C_{GK}(X,Y)$ = 0. 
\end{proof}
Note that the above characterizations for $C_{GK}(X,Y)$ and $C_{W}(X,Y)$
are of fundamentally different nature from their original characterizations
and provide important insights into the understanding of shared information.
Moreover, from a practical standpoint, these characterizations also
play a role in finding the minimum communication cost for networks
when the cost of transmission on each link is a non-linear function
of the rate as illustrated in \cite{ITW_CI}.

\section{Conclusion}

In this paper we derived single letter information theoretic characterizations
for the lossy generalizations of the two most prevalent notions of
CI due to Wyner and due to Gács and Körner. These generalizations
allow us to extend the theoretical interpretation underlying their
original definitions to sources with infinite entropy (eg. continuous
random variables). We use these information theoretic characterizations
to derive the CI of bivariate Gaussian random variables. We finally
showed that the operating points associated with the two notions of
CI arise as extreme special cases of a broader framework, that involves
the tradeoff between the the total transmit versus the receive rate
in the Gray-Wyner network. For the lossless setting, single letter
information theoretic characterization for the tradeoff curve was
established. Using the properties of the tradeoff curve, alternate
characterizations under a common framework were derived for the two
notions of CI.

\appendix
%dummy comment inserted by tex2lyx to ensure that this paragraph is not empty

\section*{Appendix A: Proof of Theorem 1 for the general setting }

In this appendix, we extend the proof of Theorem 1 for general sources
and distortion measures. Here, we do \textit{not} assume that every
point on the intersection of the Gray-Wyner region and the Pangloss
plane satisfies (\ref{eq:GW_Lossy_Region}) with equality for some
joint density $P(X,Y,U,\hat{X},\hat{Y})$. We note that an equivalent
definition of Wyner's lossy CI is given by the following. For any
$\epsilon>0$, let $R_{0}^{min}(D_{1},D_{2},\epsilon)$ be defined
as: 
\begin{eqnarray}
R_{0}^{min}(D_{1},D_{2},\epsilon) & = & \inf R_{0}
\end{eqnarray}
over all points $(R_{0},R_{1},R_{2})$ satisfying:
\begin{eqnarray}
(R_{0},R_{1},R_{2}) & \in & \mathcal{R}_{GW}(D_{1},D_{2})\nonumber \\
R_{0}+R_{1}+R_{2} & \leq & R_{X,Y}(D_{1},D_{2})+\epsilon
\end{eqnarray}
Then,
\begin{equation}
C_{W}(X,Y;D_{1},D_{2})=\lim_{\epsilon\rightarrow0}R_{0}^{min}(D_{1},D_{2},\epsilon)
\end{equation}
In the following Lemma, we derive upper and lower bounds to $C_{W}(X,Y;D_{1},D_{2})$
in terms of $\epsilon$. 
\begin{lem}
Let $\epsilon>0$ be given. Then, an upper bound to $C_{W}(X,Y;D_{1},D_{2})$
is:
\begin{equation}
C_{W}(X,Y;D_{1},D_{2})\leq\inf I(X,Y;U)
\end{equation}
where the infimum is over all joint densities $P(X,Y,\hat{X},\hat{Y},U)$
satisfying:
\begin{eqnarray}
I(X,Y;\hat{X},\hat{Y}) & \leq & R_{X,Y}(D_{1},D_{2})+\epsilon\nonumber \\
I(X,Y;U|\hat{X},\hat{Y}) & = & 0\nonumber \\
I(\hat{X};\hat{Y}|U) & = & 0\nonumber \\
E(d_{X}(X,\hat{X})) & \geq & D_{1}\nonumber \\
E(d_{Y}(Y,\hat{Y})) & \geq & D_{2}
\end{eqnarray}
We denote this upper bound by \textup{$C_{W}^{UB}(D_{1},D_{2},\epsilon)$. }

A lower bound to $C_{W}(X,Y;D_{1},D_{2})$ is:
\begin{equation}
C_{W}(X,Y;D_{1},D_{2})\geq\inf I(X,Y;U)
\end{equation}
where the infimum is over all joint densities $P(X,Y,\hat{X},\hat{Y},U)$
satisfying:
\begin{eqnarray}
I(X,Y;\hat{X},\hat{Y}) & \leq & R_{X,Y}(D_{1},D_{2})+\epsilon\nonumber \\
I(X,Y;U|\hat{X},\hat{Y}) & \leq & \epsilon\nonumber \\
I(\hat{X};\hat{Y}|U) & \leq & \epsilon\nonumber \\
E(d_{X}(X,\hat{X})) & \geq & D_{1}\nonumber \\
E(d_{Y}(Y,\hat{Y})) & \geq & D_{2}
\end{eqnarray}
We denote this lower bound by \textup{$C_{W}^{LB}(D_{1},D_{2},\epsilon)$.}\end{lem}
\begin{proof}
The proof follows using very similar arguments to that in the proof
of Theorem \ref{thm:main}. Hence, we omit the details here to avoid
repetition. 
\end{proof}
Observe that the proof of Theorem \ref{thm:main} for the general
setting follows once we show that $C_{W}^{UB}(D_{1},D_{2},\epsilon)$
and $C_{W}^{LB}(D_{1},D_{2},\epsilon)$ are continuous at $\epsilon=0$.
The following Lemma sheds light on the the continuity of these quantities
at $\epsilon=0$.
\begin{lem}
Let $(D_{1},D_{2})$ be a pair of distortions at which there exists
at least one distribution $P(\hat{X},\hat{Y}|X,Y)$ that is RD-optimal
in achieving $R_{X,Y}(D_{1},D_{2})$. Then:
\begin{equation}
C_{W}(X,Y;D_{1},D_{2})=C_{W}^{LB}(D_{1},D_{2},0)=C_{W}^{UB}(D_{1},D_{2},0)
\end{equation}
However, if there exists no RD-optimal distribution, i.e., there only
exist distributions that can infinitesimally approach $R_{X,Y}(D_{1},D_{2})$,
then: 
\begin{eqnarray}
C_{W}(X,Y;D_{1},D_{2}) & = & \lim_{\epsilon\rightarrow0}C_{W}^{LB}(D_{1},D_{2},\epsilon)\nonumber \\
 & = & \lim_{\epsilon\rightarrow0}C_{W}^{UB}(D_{1},D_{2},\epsilon)
\end{eqnarray}
\end{lem}
\begin{proof}
To prove this Lemma, we employ techniques very similar to the ones
used by Wyner in \cite{Wyner_CI}. We further restrict ourselves to
discrete random variables for simplicity. However, the arguments can
be easily extended to well-behaved continuous random variables and
distortion measures using standard techniques. 

We first show that the quantity $C_{W}^{LB}(D_{1},D_{2},\epsilon)$
is a convex function of $\epsilon$ for all $\epsilon>0$. Towards
proving this, define a generalized version of $C_{W}^{LB}(D_{1},D_{2},\epsilon)$
as:
\begin{equation}
C_{W}^{LB}(D_{1},D_{2},\epsilon_{1},\epsilon_{2},\epsilon_{3})=\inf I(X,Y;U)\label{eq:inf_prob_all_3}
\end{equation}
where the infimum is over all joint densities $P(X,Y,\hat{X},\hat{Y},U)$
satisfying:
\begin{eqnarray}
I(X,Y;\hat{X},\hat{Y}) & \leq & R_{X,Y}(D_{1},D_{2})+\epsilon_{1}\nonumber \\
I(X,Y;U|\hat{X},\hat{Y}) & \leq & \epsilon_{2}\nonumber \\
I(\hat{X};\hat{Y}|U) & \leq & \epsilon_{3}\nonumber \\
E(d_{X}(X,\hat{X})) & \geq & D_{1}\nonumber \\
E(d_{Y}(Y,\hat{Y})) & \geq & D_{2}\label{eq:consr_all_3}
\end{eqnarray}
Particularly, we show that $C_{W}^{LB}(D_{1},D_{2},\epsilon_{1},\epsilon_{2},\epsilon_{3})$
is convex with respect to $\epsilon_{i}$ for any fixed values of
$\epsilon_{j}$ and $\epsilon_{k}$, $i,j,k\in\{1,2,3\}$. 

Let $\epsilon_{2}>0$ and $\epsilon_{3}>0$ be fixed. Let $\epsilon_{11}>0$
and $\epsilon_{12}>0$ be two values for $\epsilon_{1}$ and let the
corresponding optimizing distributions for (\ref{eq:inf_prob_all_3})
be $P_{1}(X,Y,\hat{X},\hat{Y},U)$ and $P_{2}(X,Y,\hat{X},\hat{Y},U)$,
respectively. Now consider $\epsilon_{1}=\theta\epsilon_{11}+(1-\theta)\epsilon_{12}$,
for some $0<\theta<1$. It is easy to check that the joint distribution
that takes value $P_{1}$ with probability $\theta$ and value $P_{2}$
with probability $1-\theta$, denoted hereafter by $P_{\theta}$,
satisfies all the constraints in (\ref{eq:consr_all_3}) for $\epsilon_{1},$$\epsilon_{2}$
and $\epsilon_{3}$. Next, consider the following series of inequalities:
\begin{eqnarray*}
 &  & C_{W}^{LB}(D_{1},D_{2},\epsilon_{1},\epsilon_{2},\epsilon_{3})\\
 &  & \leq I^{P_{\theta}}(X,Y;U)\\
 &  & =\theta I^{P_{1}}(X,Y;U)+(1-\theta)I^{P_{2}}(X,Y;U)\\
 &  & =\theta C_{W}^{LB}(D_{1},D_{2},\epsilon_{11},\epsilon_{2},\epsilon_{3})\\
 &  & \,\,\,\,\,\,\,+(1-\theta)C_{W}^{LB}(D_{1},D_{2},\epsilon_{12},\epsilon_{2},\epsilon_{3})
\end{eqnarray*}
where $I^{P}(\cdot,\cdot)$ denotes the mutual information with respect
to joint density $P$. This proves convexity of $C_{W}^{LB}(D_{1},D_{2},\epsilon_{1},\epsilon_{2},\epsilon_{3})$
with respect to $\epsilon_{1}$ for fixed $\epsilon_{2}$ and $\epsilon_{3}$.
Similar arguments lead to the conclusion that $C_{W}^{LB}(D_{1},D_{2},\epsilon_{1},\epsilon_{2},\epsilon_{3})$
is convex with respect to $(\epsilon_{1},\epsilon_{2},\epsilon_{3})$
when $\epsilon_{1}>0$, $\epsilon_{2}>0$ and $\epsilon_{3}>0$. Hence,
$C_{W}^{LB}(D_{1},D_{2},\epsilon)$ and $C_{W}^{UB}(D_{1},D_{2},\epsilon)$
are convex and continuous for all $\epsilon>0$. 

To prove that $C_{W}^{LB}(D_{1},D_{2},\epsilon_{1},\epsilon_{2},\epsilon_{3})$
is continuous at the origin, we first consider continuity with respect
to $\epsilon_{2}$ for fixed $\epsilon_{1}$ and $\epsilon_{3}$.
Let $\epsilon_{1},\epsilon_{2},\epsilon_{3}>0$ and let $P(X,Y,\hat{X},\hat{Y},U)$
be the joint density that achieves the optimum for $C_{W}^{LB}(D_{1},D_{2},\epsilon_{1},\epsilon_{2},\epsilon_{3})$.
It is sufficient for us to prove that there exists a joint density
$Q(X,Y,\hat{X},\hat{Y},U)$ that satisfies (\ref{eq:consr_all_3})
with $\epsilon_{2}=0$, and for which $I^{Q}(X,Y;U)$ is within $\delta(\epsilon_{2})$
from $C_{W}^{LB}(D_{1},D_{2},\epsilon_{1},\epsilon_{2},\epsilon_{3})$,
for some $\delta(\epsilon_{2})\rightarrow0$ as $\epsilon_{2}\rightarrow0$.
We construct the joint density $Q(\cdot)$ as follows:
\[
Q(X,Y,\hat{X},\hat{Y},U)=P(\hat{X},\hat{Y})P(X,Y|\hat{X},\hat{Y})P(U|\hat{X},\hat{Y})
\]
Observe that all the conditions in (\ref{eq:consr_all_3}) with $\epsilon_{2}=0$
are satisfied by $Q(\cdot)$, as $I^{Q}(X,Y;U|\hat{X},\hat{Y})=0$.
We need to show that $|I^{P}(X,Y;U)-I^{Q}(X,Y;U)|\leq\delta(\epsilon_{2})$
for some $\delta(\epsilon_{2})\rightarrow0$ as $\epsilon_{2}\rightarrow0$.
Towards proving this result, we have:
\begin{eqnarray*}
\epsilon_{2} & = & I^{P}(X,Y;U|\hat{X},\hat{Y})\\
 & \geq^{(a)} & \sup|P(X,Y,U|\hat{X},\hat{Y})-Q(X,Y,U|\hat{X},\hat{Y})|
\end{eqnarray*}
where $(a)$ follows from Pinsker's inequality \cite{Cover-book}
and the supremum is over all possible subsets of alphabets of $(X,Y,U,\hat{X},\hat{Y})$.
The above inequalities state that the joint densities $P(X,Y,\hat{X},\hat{Y},U)$
and $Q(X,Y,\hat{X},\hat{Y},U)$ have a total variation smaller than
$\epsilon_{2}$. Therefore, as $\epsilon_{2}\rightarrow0$, $P(\cdot)\rightarrow Q(\cdot)$.
As conditional entropy is continuous in the total variation distance
of the corresponding random variables, there exists a $\delta(\epsilon_{2})$
such that $|I^{P}(X,Y;U)-I^{Q}(X,Y;U)|\leq\delta(\epsilon_{2})$ and
$\delta(\epsilon_{2})\rightarrow0$ as $\epsilon_{2}\rightarrow0$.
This proves that $\lim_{\epsilon_{2}\rightarrow0}C_{W}^{LB}(D_{1},D_{2},\epsilon_{1},\epsilon_{2},\epsilon_{3})$
$=C_{W}^{LB}(D_{1},D_{2},\epsilon_{1},0,\epsilon_{3})$. The proof
for continuity with respect to $\epsilon_{3}$ at origin follows in
very similar lines to the proof of Theorem 4.4 in \cite{Wyner_CI}.
Hence, we omit the details here. This leads to the conclusion that:
\begin{eqnarray*}
C_{W}(X,Y;D_{1},D_{2}) & = & \lim_{\epsilon\rightarrow0}C_{W}^{LB}(D_{1},D_{2},\epsilon)\\
 & = & \lim_{\epsilon\rightarrow0}C_{W}^{UB}(D_{1},D_{2},\epsilon)
\end{eqnarray*}
The proof of this Lemma follows directly by observing that, if there
exists a joint density $P(\hat{X},\hat{Y}|X,Y)$ that achieves $R_{X,Y}(D_{1},D_{2})$,
then both the above limits converge to $C_{W}^{UB}(D_{1},D_{2},0)$,
which is same as the information theoretic characterization in (\ref{eq:mt_1}). 
\end{proof}

\section*{Appendix B: Proof of Corollary 1\label{app:Proof_Lemma_1}}
\begin{proof}
Our objective is to prove that the optimizing distribution in Theorem
\ref{thm:main} satisfies (\ref{eq:Lem_add_Mark_Prop}). We begin
with an auxiliary property of the RD-optimal conditional distribution
$P(\hat{X}^{*},\hat{Y}^{*}|X,Y)$. Recall that the RD optimal conditional
distribution minimizes $I(X,Y;\hat{X}^{*},\hat{Y}^{*})$ over all
joint distributions that satisfy the distortion constraints, $E(d_{X}(X,\hat{X}^{*}))\leq D_{1}$
and $E(d_{Y}(Y,\hat{Y}^{*}))\leq D_{2}$. Hence, it follows using
standard arguments \cite{Berger_book} that, for every distortion
pair $(D_{1},D_{2})$, the RD-optimal conditional distribution, $P(\hat{X}^{*},\hat{Y}^{*}|X,Y)$,
also minimizes the following Lagrangian for some positive constants
$\mu_{1},\mu_{2}$ and positive valued function $\lambda(x,y)$ defined
for all $x\in\mathcal{X}$,$y\in\mathcal{Y}$:
\begin{eqnarray}
L & = & I(X,Y;\hat{X}^{*},\hat{Y}^{*})\nonumber \\
 &  & -\mu_{1}E(d_{X}(X,\hat{X}^{*}))-\mu_{2}E(d_{Y}(Y,\hat{Y}^{*}))\\
 &  & -\int\lambda(x,y)dxdy\int d\hat{x}^{*}d\hat{y}^{*}P(\hat{x}^{*},\hat{y}^{*}|x,y)\nonumber 
\end{eqnarray}
Upon differentiating the above Lagrangian with respect to $P(\hat{x}^{*},\hat{y}^{*}|x,y)$
and setting it to zero leads to the necessary conditions for optimality
of the joint distribution. Some routine steps and simplifications
lead to the conclusion that the joint distribution $P(X,Y,\hat{X}^{*},\hat{Y}^{*})$
must satisfy the following:
\begin{eqnarray}
P(\hat{X}^{*},\hat{Y}^{*}|X,Y) & = & \lambda(X,Y)P(\hat{X}^{*},\hat{Y}^{*})\\
 &  & \times\exp(d_{X}(X,\hat{X}^{*}))\exp(d_{Y}(Y,\hat{Y}^{*}))\nonumber 
\end{eqnarray}
It follows that the conditional density $P(X,Y|\hat{X}^{*},\hat{Y}^{*})$
satisfies:

\begin{equation}
P(X,Y|\hat{X}^{*},\hat{Y}^{*})=\phi(X,Y)\exp(d_{X}(X,\hat{X}^{*}))\exp(d_{Y}(Y,\hat{Y}^{*}))
\end{equation}
where $\phi(X,Y)=\lambda(X,Y)P(X,Y)$. 

Next recall that the optimizing distribution in Theorem \ref{thm:main}
satisfies the following two Markov conditions:
\begin{eqnarray*}
(X,Y) & \leftrightarrow(\hat{X}^{*},\hat{Y}^{*})\leftrightarrow & U
\end{eqnarray*}
\begin{equation}
\hat{X}^{*}\leftrightarrow U\leftrightarrow\hat{Y}^{*}
\end{equation}
Hence the optimizing joint distribution can be rewritten as:
\begin{eqnarray}
P(X,Y,\hat{X}^{*},\hat{Y}^{*},U) & = & P(U)P(\hat{X}^{*}|U)P(\hat{Y}^{*}|U)\nonumber \\
 &  & \times P(X,Y|\hat{X}^{*},\hat{Y}^{*})\\
 & = & \phi_{1}(X,\hat{X}^{*},U)\phi_{2}(Y,\hat{Y}^{*},U)\phi(X,Y)\nonumber 
\end{eqnarray}
where $\phi_{1}(X,\hat{X}^{*},U)=P(U)P(\hat{X}^{*}|U)\exp(d_{X}(X,\hat{X}^{*}))$
and $\phi_{2}(Y,\hat{Y}^{*},U)=P(\hat{Y}^{*}|U)\exp(d_{Y}(Y,\hat{Y}^{*}))$.
Hence, it follows that:
\begin{equation}
P(\hat{X}^{*},\hat{Y}^{*},U|X,Y)=\phi_{1}(X,\hat{X}^{*},U)\phi_{2}(Y,\hat{Y}^{*},U)\lambda(X,Y)
\end{equation}
which implies that the Markov conditions in (\ref{eq:Lem_add_Mark_Prop})
must be satisfied, proving the Lemma. 
\end{proof}

\section*{Appendix C: Shannon Lower Bound for Vectors}

In this appendix, we review some of the definitions and results pertinent
to Shannon lower bounds for vectors of random variables. We refer
to \cite{Berger_book} (section 4.3.1) for further details on Shannon
lower bound and its properties. 

Let $\mathbf{X}$ be an n-dimensional random variable distributed
according to $p(\mathbf{X})$, and let $\rho_{i}(x_{i},\hat{x_{i}})$
$\forall i\in\{1,\ldots,N\}$ be any well defined difference distortion
measures, i.e., $\rho_{i}(x,\hat{x})=\rho_{i}(x-\hat{x})$. Let $R_{\mathbf{X}}(\mathbf{D})$
be the rate-distortion function of $\mathbf{X}$, with respect to
the given distortion measures, i.e.:
\begin{equation}
R_{\mathbf{X}}(\mathbf{D})=\inf_{P(\hat{\mathbf{X}}|\mathbf{X}):E\left[\rho_{i}(x_{i},\hat{x_{i}})\right]\leq D_{i}\,\,\forall i}\,\, I(\mathbf{X},\mathbf{\hat{X}})
\end{equation}
Then the Shannon lower bound to $R_{\mathbf{X}}(\mathbf{D})$, denoted
by $R_{\mathbf{X}}^{L}(\mathbf{D})$, is given by:
\begin{eqnarray}
R_{\mathbf{X}}^{L}(\mathbf{D}) & = & H(\mathbf{X})\nonumber \\
 &  & -\sup_{s_{1},\ldots,s_{n}<0}\sum_{i=1}^{N}\Bigl\{ s_{i}D_{i}-\log\int e^{s_{i}\rho_{i}(z_{i})}dz_{i}\Bigr\}\nonumber \\
 & = & H(\mathbf{X})-\sum_{i=1}^{N}\max_{g_{i}\in G_{i}(D_{i})}H(g_{i})\label{eq:Shannon_LB_Defn}
\end{eqnarray}
where, $G_{i}(D_{i})$ denotes the set of all joint distributions
such that:
\begin{equation}
\int\rho_{i}(z_{i})g(z_{i})dz_{i}\leq D_{i}\label{eq:Shannon_LB_Cond}
\end{equation}
The above derivation is a direct extension of the derivation in Section
4.3.1 in \cite{Berger_book}, to vectors of random variables. It is
easy to verify that the distribution $g_{i}$ that achieves the maximum
in (\ref{eq:Shannon_LB_Defn}) is given by:
\begin{equation}
g_{i}(z)=\frac{e^{s_{i}\rho(z)}}{\int e^{s_{i}\rho(z)}dz}\label{eq:gi_degn}
\end{equation}
where, $s_{i}$ is such that:
\begin{equation}
\int g_{i}(z)\rho_{i}(z)dz=D_{i}\label{eq:di_cond}
\end{equation}

Shannon showed that $R_{\mathbf{X}}^{L}(\mathbf{D})\leq R_{\mathbf{X}}(\mathbf{D})$
always holds (see \cite{Berger_book} for details). The following
lemma states the necessary and sufficient conditions for $R_{\mathbf{X}}^{L}(\mathbf{D})=R_{\mathbf{X}}(\mathbf{D})$.
\begin{lem}
$R_{\mathbf{X}}^{L}(\mathbf{D})=R_{\mathbf{X}}(\mathbf{D})$ iff the
distribution of $\mathbf{X}$ can be expressed as: 
\begin{equation}
p(\mathbf{x})=\int q(\mathbf{\hat{x}})\prod_{i}^{n}g_{i}(x_{i}-\hat{x}_{i})d\mathbf{\hat{x}}\label{eq:Shannon_LB_Eq_Cond}
\end{equation}
i.e., $\mathbf{X}$ can be expressed as the sum of two statistically
independent random vectors, $\hat{\mathbf{X}}$ and $\mathbf{Z}$,
where $\mathbf{Z}$ is distributed according to:
\begin{equation}
g(\mathbf{Z})=\prod_{i}^{n}g_{i}(z_{i})\label{eq:gz_defn}
\end{equation}
where \textup{$g_{i}(z)$ is given by (\ref{eq:gi_degn}).}\end{lem}
\begin{proof}
Direct extension of Theorem 4.3.1 in \cite{Berger_book}.
\end{proof}
It follows from the above Lemma that, if Shannon lower bound is tight,
$\mathbf{\hat{X}}$ is the RD-optimal reconstruction and the RD-optimal
backward channel from $\mathbf{\hat{X}}$ to $\mathbf{X}$ is additive
and can be written as:
\begin{equation}
\mathbf{X}=\hat{\mathbf{X}}+\mathbf{Z}
\end{equation}
where $\mathbf{Z}\sim\prod_{i}^{n}g_{i}(Z_{i})$. Therefore, if Shannon
lower bound is tight, the components of $\mathbf{X}$ are independent
given $\hat{\mathbf{X}}$, i.e., the joint density $p(\mathbf{X},\hat{\mathbf{X}})$
is of the form:
\begin{equation}
p(\mathbf{X},\hat{\mathbf{X}})=q(\hat{\mathbf{X}})\prod_{i}^{n}p_{i}(X_{i}|\hat{X}_{i})\label{eq:Shannon_LB_Lemma_Cond}
\end{equation}

\bibliographystyle{ieeetr}
\bibliography{Journal_Bibtex}

\end{document}